%% file: main.tex
\documentclass[11pt]{article}
\input{macros}

\title{A New Dynamic Algorithm for Densest Subhypergraphs\footnote{An extended abstract of this paper appears in The Web Conference (formerly WWW) 2022}}

\author{
Suman K. Bera\thanks{Katana Graph, USA. The main work was done while the author was at UC Santa Cruz and was supported by NSF TRIPODS grant CCF-1740850, NSF CCF-1813165, CCF-1909790, and ARO Award W911NF1910294.} \and Sayan Bhattacharya\thanks{University of Warwick, UK. Supported by Engineering and Physical Sciences Research Council, UK (EPSRC) Grant EP/S03353X/1.} \and Jayesh Choudhari\footnotemark[3] \and
Prantar Ghosh\thanks{Dartmouth College, USA. Supported in part by NSF
under award CCF-1907738.}
}
\date{}

\begin{document}

\maketitle

\begin{abstract}
 Computing a dense subgraph  is a fundamental problem in graph mining, with a diverse set of applications ranging from electronic commerce to community detection in social networks. In many of these applications, the underlying context is  better modelled as a weighted hypergraph that keeps evolving with time. 
 
 This motivates the problem of  maintaining the densest subhypergraph of a weighted hypergraph in a {\em dynamic setting}, where the input keeps changing via a sequence of updates (hyperedge insertions/deletions). Previously, the only known algorithm  for this problem was due to Hu et al.~\cite{hu2017maintaining}. This algorithm  worked only on unweighted hypergraphs, and had an approximation ratio of   $(1+\epsilon)r^2$ and an update time of $O(\text{poly} (r, \log n))$, where $r$ denotes the maximum rank of the input across all the updates. 
 
 We obtain a new algorithm for this problem, which works even when the input hypergraph is weighted. Our algorithm  has a significantly improved (near-optimal) approximation ratio of $(1+\epsilon)$ that is independent of $r$, and a similar update time of $O(\text{poly} (r, \log n))$. It is the first $(1+\eps)$-approximation algorithm even for the special case of weighted simple graphs.

To complement our theoretical analysis, we perform experiments
with our dynamic algorithm on large-scale, real-world data-sets.
Our algorithm significantly outperforms the state of the art \cite{hu2017maintaining}
both in terms of accuracy and efficiency.
\end{abstract}

\input{intro}

\input{prelims}

\input{weightscale}

\input{reductionunweighted}

\input{udshp}

\input{experiments}

\input{conclusions}

\bibliographystyle{alpha}
\bibliography{refs}

\appendix

\input{additional_expts}

\end{document}

%% file: macros.tex
\usepackage[margin=1in]{geometry}
\usepackage{amssymb}
\usepackage{amsmath,xspace,verbatim,paralist,enumitem}
\usepackage{amsthm}
\usepackage{xcolor}

\usepackage[colorlinks, linkcolor=blue, citecolor=blue]{hyperref} 
\usepackage{url} 
\usepackage{thm-restate}
\usepackage[utf8]{inputenc}
\usepackage[T1]{fontenc}
\usepackage{caption, multicol}
\usepackage{booktabs}
\usepackage{graphicx}

\usepackage{algorithm}
\usepackage[noend]{algorithmic}

\usepackage[nameinlink]{cleveref}

\usepackage[parfill]{parskip}

\usepackage{enumitem}

\usepackage[font=small,skip=1pt]{caption}

\usepackage{subcaption}

\crefformat{footnote}{#2\footnotemark[#1]#3}

\DeclareCaptionType{Data Structure}

\newtheorem{theorem}{Theorem}[section]
\newtheorem{lemma}[theorem]{Lemma}

\newtheorem{corollary}[theorem]{Corollary}

\newtheorem{observation}[theorem]{Observation}

\newtheorem{fact}[theorem]{Fact}

\theoremstyle{definition}

\newcommand{\mypar}[1]{\noindent{\sffamily\bfseries #1.}~}

\newcommand{\etal}{{et al.}\xspace}

\newcommand{\eps}{\epsilon}
\newcommand{\hrho}{\hat{\rho}}
\newcommand{\Gunw}{G^{\text{unw}}}

\newcommand{\hf}{\hat{f}}
\newcommand{\hl}{\hat{\ell}}
\newcommand{\hD}{\hat{D}}
\newcommand{\din}{d_{in}}

\newcommand{\tD}{\widetilde{D}}

\newcommand{\NN}{\mathbb{N}}
\newcommand{\RR}{\mathbb{R}}

\newcommand{\exact}{\textsc{Exact}\xspace}
\newcommand{\hwc}{\textsc{HWC}\xspace}
\newcommand{\udshp}{\textsc{Udshp}\xspace}
\newcommand{\wdshp}{\textsc{Wdshp}\xspace}

\newcommand{\EE}{\mathbb{E}}

\newcommand{\maxweight}{w_{\max}}

\newcommand{\Gp}{G_{q(\trho)}}

\newcommand{\trho}{\widetilde{\rho}}

\newcommand{\tagu}{\texttt{tag-ask-ubuntu}\xspace}
\newcommand{\tagm}{\texttt{tag-math-sx}\xspace}
\newcommand{\tagsf}{\texttt{tag-stack-overflow}\xspace}
\newcommand{\cmag}{\texttt{Coauth-MAG}\xspace}
\newcommand{\dawn}{\texttt{DAWN}\xspace}

\newcommand{\dblpall}{\texttt{DBLP-All}\xspace}

%% file: intro.tex
\section{Introduction}
In the weighted densest subhypergraph (WDSH) problem, we are given a weighted hypergraph $G=(V,E,w)$ as input, where $w:E\to \RR^+$ is a weight function. The {\em density} of any subset of vertices $U \subseteq V$ in $G$ is defined as $\rho_G(U):= (\sum_{e\in E[U]} w_e)/|U|$, where $E[U]$ is the set of hyperedges induced by $U$ on~$G$. The goal is to find a subset of vertices $U \subseteq V$ in $G$ with maximum density.

We consider the {\em dynamic} WDSH problem, where the input hypergraph $G$ keeps changing via a sequence of {\em updates}. Each update either deletes a hyperedge from $G$, or inserts a new hyperedge $e$ into $G$ and specifies its weight $w_e$. In this setting, the {\em update time} of an algorithm refers to the time it takes to handle an update in~$G$. We want to design an algorithm that maintains a (near-optimal) densest subhypergraph in $G$ with  small update time. 

The {\em rank of a hyperedge} $e$  is  the number of vertices incident on~$e$. The {\em rank of a hypergraph} is the maximum rank among all its hyperedges. Let $r$ denote an upper bound on the {\em rank} of  of the input hypergraph throughout the sequence of updates. Let $n$ be the number of nodes and $m$ be an upper bound on the number of hyperedges over the sequence of updates. Our main result is summarized below.

\begin{theorem}(Informal)
\label{thm:main:informal}
There is a randomized  $(1+\epsilon)$-approximation algorithm for the dynamic WDSH problem with $O(r^2 \cdot \text{polylog} (n, m))$ worst case update time, for every sufficiently small constant $\epsilon > 0$. 
\end{theorem}
Note that a  naive approach for  this problem would be to run a static algorithm from scratch after every update, which leads to  $\Omega(r \cdot (n+m))$ update time.  As $r$ is a small constant in most practical applications, the update time of our dynamic algorithm is {\em exponentially smaller} than the update time of this naive approach.

\subsection{Perspective}

Computing a dense subgraph is a fundamental primitive in graph mining~\cite{goldberg1984finding, LangA07, BahmaniKV12, tsourakakis2013denser, BalalauBCGS15, GionisT15, tsourakakis2015k, BoobGPSTWW20}. 
Over the course of past several decades, it has been found to be useful in a range  of different contexts, such as {\em community detection}~\cite{ChenS12} and {\em piggybacking} on social networks~\cite{GionisJLSW13}, {\em spam detection} in the web~\cite{DGibsonKT05}, {\em graph compression}~\cite{FederM95}, {\em expert team formation}~\cite{BonchiGKV14},  {\em real-time story identification}~\cite{AngelKSS12}, {\em computational biology}~\cite{SahaHKRZ10} and {\em electronic commerce}~\cite{LangA07}. There are three features that stand out from this diverse list of applications and motivate us to study the more general dynamic WDSH problem. { (I) Real-world networks are often {\em dynamic}, in the sense that they change  over time.} { (II) The underlying real-world context is often easier to capture by making the graph edges {\em weighted}.} { (III) It is often more beneficial to model the underlying network as a {\em hypergraph} rather than a standard graph.} 

In order to appreciate the significance of these three features, consider two concrete real-world examples.

\mypar{Example 1: Real-time story identification}The wide popularity of social media yields overwhelming  activity by millions of users at all times in the form of, say, tweets, status updates, or blog posts. These are often related to important current events or stories that one might seek to identify in real time. For instance, consider the recent Israel-Palestine conflict in May 2021. After the outbreak of the conflict, multiple incidents occurred in quick succession that are important to be promptly identified. An efficient technique for real-time story identification is focusing on certain ``entities'' associated with a story, e.g., famous personalities, places, organizations, products, etc. They consistently appear together in the numerous posts on the related story. In the example of the Israel-Palestine conflict, countless online posts have turned up about the several events, many of which feature the same small set of entities, e.g., {\em Israel}, {\em Palestine}, {\em Hamas}, {\em Gaza}, {\em Sheikh Jarrah}, and {\em airstrike}, or subsets thereof. This correlation can be leveraged: the set of all possible real-world entities (which can be billions) represented by nodes, with an edge connecting each pair iff they appear together in a post, define a graph that changes dynamically over time; maintaining a dense subgraph of this network helps us to identify the group of most strongly-related entities (in the example above, this group might be \{{\em Hamas}, {\em Gaza}, {\em Sheikh Jarrah}, {\em airstrike}\}), and in turn, the trending story \cite{AngelKSS12}.

Note the significance of feature (I) here: the number of posts keeps growing rapidly, thus dynamically modifying the underlying graph. Further, a large number of posts gets deleted over time. This is often driven by the proliferation of {\em fake news} and its eventual removal upon detection. Also notice that feature (II) is crucial for this task. Every minute, millions of entities get mentioned in a small number of posts. The few entities in the story of interest, however, collectively appear in a {\em massive} number of posts.
Therefore, to make them stand out, we can assign to the graph edges {\em weights} proportional to the number of posts they represent. Thus, the densest subgraph is induced by the union of the entities in the story. Finally, observe the importance of feature (III) in this context. For a post mentioning multiple entities, instead of adding an edge between each pair of them, we can simply include all of them in a single {\em hyperedge}. The standard graph formulation creates a clique among those nodes, which makes the density of the post proportional to the number of entities mentioned. This is inaccurate for several applications. In contrast, having a single hyperedge represent a post removes this bias. The task of real-time story identification thus reduces to precisely the dynamic WDSH problem.


\mypar{Example 2: Trending Topics Identification} Consider the setting where we wish to identify a set of recently trendy topics in a website like Stack Overflow. We can model this scenario as a network where each node corresponds to a tag, and there is a hyperedge containing a set of nodes iff there is a post with the corresponding set of tags. The weight of a hyperedge represents the reach of a post, captured by, say, the number of responses it generates. The set of recently trendy topics will be given by the set of tags that form the densest subhypergraph in this network. The network is dynamic: posts are added very frequently and deletions are caused not only by their actual removal but also by our interest in only the ones that appeared (say) within the last few~days. 

Other applications of the WDSH problem include identifying a group of researchers with the most impact \cite{hu2017maintaining} and analysing spectral properties of hypergraphs \cite{ChanLTZ18}.

\mypar{Previous work}In the static setting, there is a linear programming based exact algorithm for computing the densest subgraph that runs in polynomial time. In addition, a simple linear-time greedy algorithm gives $2$-approximation for this problem~\cite{Charikar00}. 

Starting with the work of Angel et al.~\cite{AngelKSS12}, in recent years a sequence of papers have dealt with the densest subgraph problem in the dynamic setting. Epasto et al.~\cite{EpastoLS15} considered a scenario where the input graph undergoes a sequence of adversarial  edge insertions and random edge deletions, and designed a dynamic $(2+\epsilon)$-approximation algorithm  with $O(\text{poly} \log n)$ update time. In the standard (adversarial) fully dynamic setting, Bhattacharya et al.~\cite{bhattacharya2015DensestSubgraph} gave a $(4+\epsilon)$-approximation algorithm with $O(\text{poly} \log n)$ update time. This latter result was recently improved upon by Sawlani and Wang~\cite{SawlaniW20}, who obtained a $(1+\epsilon)$-approximation algorithm with $O(\text{poly} \log n)$ update time. All these results, however, hold only on {\em unweighted} simple graphs (i.e., hypergraphs with rank $2$). Our algorithm, in contrast, works for {\em weighted} rank-$r$ hypergraphs and is the first $(1+\eps)$-approximation algorithm with $O(\text{poly} \log n)$ update time even for the special case of edge-weighted simple graphs.\footnote{A vertex-weighted version for simple graphs was solved by \cite{SawlaniW20}. However, it doesn't seem to generalize to the edge-weighted case.}

For general rank-$r$ hypergraphs, the only dynamic algorithm currently known in the literature was designed by Hu et al.~\cite{hu2017maintaining}: in the fully dynamic setting, their algorithm has an approximation ratio of $(1+\epsilon)r^2$ and an {\em amortized} update time of $O(\text{poly} (r, \log n))$. In direct contrast, as summarized in Theorem~\ref{thm:main:informal}, our approximation ratio is near-optimal (and independent of $r$), and our update time guarantee holds in the {\em worst case}. Furthermore, our algorithm works even when the hyperedges in the input hypergraph have large general weights, whereas the algorithm in~\cite{hu2017maintaining} needs to assume that the input hypergraph is either unweighted or has very small weights of value at most $O(\text{poly} (r, \log n))$.

As an aside, we remark that because of the fundamental importance of this problem in various graph mining applications, efficient algorithms for computing a dense subgraph have also been designed in other related computational models such as mapreduce~\cite{BahmaniKV12,BahmaniGM14} and semi-streaming~\cite{mitzenmacher2015scalable,EsfandiariHW16,McGregorTVV15}.

\mypar{Significance of our result}Given this background, let us now emphasize three aspects of our  result as stated in \Cref{thm:main:informal}.

First, the approximation ratio of our algorithm can be made arbitrarily close to $1$, and in particular, it is independent of the rank $r$ of the input hypergraph. For example, if $r = 3$, then~\cite{hu2017maintaining} can only guarantee that in the worst case, the objective value of the solution maintained by their algorithm is at least $(100/r^2)\% \simeq 11\%$ of the optimal objective value. In contrast, for {\em any} $r$, we can guarantee that the objective value of the solution maintained by our algorithm is always within $\simeq 99\%$ of the optimal objective value. In fact, since $r$ can be, in theory, as large as $n$, the improvement over the approximation ratio is massive.  

Second,  the update time of our algorithm is $O(r^2 \cdot \text{polylog}(n,m))$. Note that {\em any} dynamic algorithm for this problem will necessarily have an update time of $\Omega(r)$, since it takes $\Theta(r)$ time to even specify an update. It is not surprising, therefore, that the update time of~\cite{hu2017maintaining} also had a polynomial dependency on $r$. Since $r$ is a small constant in most practical applications, our update time is essentially $O(\text{polylog} (n, m))$ in these settings. 

Third, our dynamic algorithm works for weighted graphs, which, as noted above, are crucial for applications. Throughout the rest of this paper, we assume that the weight of every hyperedge is a positive integer. This is without loss of generality for the following reason. If the weights are positive real numbers, we can always {\em scale} them by the same multiplicative factor (without changing the optimal solution), so as to ensure that the minimum possible weight of any hyperedge is at least $1/\epsilon$. After this scaling, we can {\em round} up the weight of each hyperedge to its nearest integer. This rounding step changes the weight of every hyperedge by at most a multiplicative factor of $(1+\epsilon)$. Hence, if we can maintain a $(1+\epsilon)$-approximately optimal solution in this rounded instance (where every weight is an integer), then this gives us  a $(1+\epsilon)^2 \simeq (1+2\epsilon)$-approximation to the original input instance. Finally, if the weights of the hyperedges are known to be integers in the range $[1, W]$, then a naive approach would be to make $w_e$ copies of every hyperedge $e$ when it gets inserted, and maintain a near-optimal solution in the resulting unweighted hypergraph. This, however, leads to an update time of $\Theta(W)$. This is prohibitive when $W$ is large. In contrast, our algorithm has polylogarithmic update time for any $W$.
\subsection{An Overview of Our Techniques}

We obtain the result stated in Theorem~\ref{thm:main:informal} in two major steps. In the main body, we present these steps in reverse order for the ease of presentation.

{\bf Step I: Extending the result of~\cite{SawlaniW20} to the unweighted densest subhypergraph (UDSH) problem.}\footnote{In this version, every hyperedge has weight $= 1$.} Similar to~\cite{SawlaniW20},
 we take the primal-dual
approach. 
We observe that
the dual of the UDSH problem
can be interpreted as a load balancing 
problem, where each hyperedge needs
to distribute one unit of {\em load} among its
nodes in a way that minimizes the maximum load
received by a  vertex from all its incident 
edges. In an optimal 
solution, if a hyperedge $e$ assigns positive 
load to some node $v \in e$, then the total 
load on $v$ (from all its incident 
edges) cannot be larger than the total load on
some other node $u \in e \setminus \{v\}$; for otherwise, $e$ can transfer some 
of the load it puts on $v$ to $u$ 
without increasing the objective value. This gives us a ``local''
condition for optimality. We show that maintaining 
the local condition with some {\em slack} 
yields a near-optimal global solution. 
Interestingly, the
approximation guarantee we obtain from this analysis {\em does not} grow with the rank $r$ of the input hypergraph.

Moving on, the slack in the local condition allows us to do away with
fractional loads and have integral 
loads only. This means that a hyperedge 
can put its entire unit load on a 
single node rather than distribute 
it fractionally. Thus, the problem 
now reduces to a hypergraph 
``orientation'' problem: for each 
hyperedge $e$, select a node $v \in e$ as its ``head'' (think of 
``orienting'' the edge $e$ towards 
the node $v$) such that the total 
load on $v$ (i.e., the total number 
of hyperedges who select $v$ as 
their head) is at most the total load on $u$ {\em plus} some slack $\eta$, for all $u \in e \setminus \{ v\}$. The idea behind efficiently maintaining this orientation is simple: when a hyperedge $e$ is inserted, we select its least loaded vertex as the head. If this violates the local condition for some other hyperedge $e'$, then we ``rotate'' $e'$ towards (i.e., reassign its head to) its least loaded vertex, and so on and so forth. This process terminates when every hyperedge satisfies the local condition. Since the load goes down by an additive factor of at least $\eta$ along the head-nodes of this {\em affected chain} of hyperedges, the length of the chain can be at most $\hD/\eta$, where $\hD$ is the maximum load over all nodes. Our analysis shows that $\hD/\eta$ is roughly $\log n$, which implies that the number of hyperedge rotations we need to perform per insertion is roughly $\log n$. The deletions can be handled similarly. It now remains to design suitable data structures that can actually {\em detect} these violations and perform the rotations efficiently, leading to small update time.

The main difficulty in designing such  data structures is the following. Any change on the load of a node $v$ needs to be
 communicated to all the vertices in all the hyperedges $e$ that have $v$ as their head-node. But there can be
as many as $\hD=\Theta(m)$ such hyperedges,
implying a huge update time. Hence,  we use the  ``lazy round-robin informing'' 
technique of~\cite{SawlaniW20}, which leverages the slack 
$\eta$ to hold off communicating to certain vertices  until some later time, while still
maintaining the local condition. 



{\bf Step II: From unweighted to the weighted version.} Next, we use our algorithm for the dynamic UDSH problem as a subroutine to obtain one for the dynamic WDSH problem. For this, a naive approach would be to  view a weighted hypergraph as an
unweighted multi-hypergraph (by making as 
many unweighted copies of a hyperedge as 
its weight). Unfortunately, for each hyperedge 
insertion in WDSH, this might require us to perform 
$w_e$ insertions in UDSH, where $w_e$ is 
the weight of the edge. Thus, this approach will lead to an update time of 
$\Omega(\maxweight)$, where $\maxweight$ denotes the max-weight of an edge in the weighted hypergraph, or equivalently, the max-multiplicity of an edge in its unweighted version. Observe that, however, if $\maxweight$
is ``small'', then we have an efficient
algorithm. Accordingly, we show how to  scale down the weights of the hyperegdes in the input hypergraph so as to make $\maxweight$ small, while simultaneously scaling down the maximum density by a known factor. This leads to the desired  algorithm for the WDSH problem. 

At this point, the ``sparsification''
technique of~\cite{mitzenmacher2015scalable} comes in handy. Their result implies that if we sample each edge of the unweighted version of the hypergraph with probability roughly $\log n/\trho$, for a $(1\pm \eps)$-estimate $\trho$ of the max-density $\rho^*$, then the sampled hypergraph has max-density $\simeq \log n$. We can now make polylogarithmic many parallel guesses for $\trho$ in powers of $(1+\eps)$, and run over all of them to find the correct one that gives us a maximum density of  $\simeq \log n$ in the corresponding sampled hypergraph. This means that we can solve the WDSH problem by running the UDSH algorithm on the sampled hypergraphs corresponding to the guesses of $\trho$. Thus, our algorithm is efficient if each of them has small max-multiplicity. We show that as long as the guess $\trho$ is at least $\maxweight/r$ (a trivial lower bound on $\rho^*$, which roughly implies a sampling probability of at most $r\log n/\maxweight$), the max-multiplicity of an edge in the sampled hypergraph is  $O(r\log n)$ whp. This is small enough to give us an efficient algorithm for WDSH. 

\subsection{Experimental Evaluation}
\label{subsec:intro:exp}
We conduct extensive experiments to demonstrate the effectiveness of
our algorithm in both fully dynamic and insertion-only settings with weighted and unweighted hypergraphs. We test our algorithm on several real-world temporal hypergraph datasets (see~\Cref{sec:exp}). Based on our analysis, we report the following key findings. 

\mypar{Key Findings} 
\begin{inparaenum}[\bfseries (1)]
\item For unweighted fully dynamic settings, our algorithm significantly outperforms the state of the art of~\cite{hu2017maintaining} both in terms of accuracy and efficiency. For smaller datasets, at comparable accuracy with~\cite{hu2017maintaining}, our algorithm provides 7x-10x speed-up on the average update time. For larger datasets, for comparable update times, our algorithms offers 55\%-90\% improvement in accuracy over~\cite{hu2017maintaining}~(\Cref{fig:dyn_unweighted_time_max_avg_density}). 
\item A similar trend persists in the unweighted insertion-only settings as well. On average, our algorithm potentially achieves more than 80\% improvement in accuracy while being 2x-4x times faster compared to~\cite{hu2017maintaining}~(\Cref{fig:inc_unweighted_time_max_avg_error}).
\item In comparison against an LP solver for computing the exact solution, our algorithm offers remarkable scalability. In the insertion-only weighted and unweighted settings, our algorithms are on average at least a few hundred times faster in most cases, while incurring less than a few percentage points of relative error. For the dynamic weighted and unweighted settings, the gain is significant for larger datasets: at least 10x reduction in average update time while incurring less than 10\% relative error. 
For smaller datasets, we achieve 2x-5x speed-up while sacrificing less than 4\% in accuracy (Figures \ref{fig:avg_time_error_dyn_unweighted}, \ref{fig:avg_time_error_dyn_weighted}, and \ref{fig:avg_time_error_inc_unweighted}).
\end{inparaenum}


%% file: prelims.tex
\section{Preliminaries and Notations}
\label{sec:prelim}
Let us fix the notations that we use throughout the paper. Our input weighted hypergraph is always a rank-$r$ hypergraph denoted by $G=(V,E,w)$, where $w: E\to \NN$ is a weight function. We denote the number of vertices $|V|$ and hyperedges $|E|$ (or an upper bound on it) by $n$ and $m$ respectively. The maximum weight of a hyperedge in $G$ is given by $\maxweight(G) := \max_{e\in E} ~w_e$. The {\em multiplicity} of an edge in a multi-hypergraph is its number of copies in the hypergraph. For a subset of nodes $U\subseteq V$, denote its density in $G$ by \[\rho_G(U):=\frac{\sum_{e\in E[U]} w_e}{|U|}\,,\]
where $E[U]$ is the set of hyperedges induced by $U$ on $G$. If the hypergraph is unweighted, then the density of $U$ is simply $\rho_G(U)=|E[U]|/|U|$. We denote the maximum density of $G$ by \[\rho^*(G):= \max_{U\subseteq V} \rho_G(U)\]
We drop the argument $G$ from each of the above when the hypergraph is clear from the context.

We use the shorthands WDSH and UDSH for weighted and unweighted densest subhypergraph respectively. For the dynamic WDSH and UDSH problems, we get two types of queries: (a) max-density query, which asks the value of the maximum density over all subsets of nodes of the hypergraph, and (b) densest-subset query, which asks for a subset of nodes with the maximum density. We say an algorithm {\em maintains} an $\alpha$-approximation $(\alpha>1)$ to either of these problems if it answers every max-density query with a value that lies in  $[\rho^*/\alpha, \rho^*]$ and every densest-subset query with a subset of nodes whose density lies in $[\rho^*/\alpha, \rho^*]$. 

Given any weighted hypergraph $G$, we denote by $\Gunw$ the unweighted multi-hypergraph obtained by replacing each edge $e$ with weight $w_e$ by $w_e$ many unweighted copies of $e$. Note that $G$ and $\Gunw$ are equivalent in terms of subset densities.

We say that a statement holds {\em whp} (with high probability) if it holds with probability at least $1-1/\text{poly}(n)$.

We use the following version of the Chernoff bound.
\begin{fact}(Chernoff bound)\label{fact:chernoff}
Let $X$ be a sum of mutually independent indicator random variables. Let $\mu$ and $\delta$ be real numbers such that $\EE[X] \le \mu$ and $0 \le \delta \le 1$. Then,
  $\Pr\left[ |X-\mu| \ge \delta\mu \right] \le \exp\left( - \mu \delta^2 / 3 \right)$.
\end{fact}

In \Cref{sec:weightedmain}, we show how we can use an algorithm for UDSH as a subroutine to solve the WDSH problem, thus proving our main result in \Cref{thmmain} (formal version of \Cref{thm:main:informal}). In \Cref{sec:unweighted}, we design the algorithm for the UDSH problem. Finally, in \Cref{sec:exp}, we give an account of our experimental evaluations.

%% file: weightscale.tex
\section{Reduction to the Unweighted Case}\label{sec:weightedmain}

In this section, we show that we can use an algorithm for the dynamic UDSH problem to obtain one for the dynamic WDSH problem while incurring only a small increase in the update and query times. In \Cref{sec:wtscale}, we show the use of randomization to scale down the edge weights so as to obtain a hypergraph whose maximum edge-weight is small and the max-density is scaled down by a known factor. 
Then, in \Cref{sec:wdshp}, we use this method to design and analyse our algorithm for WDSH.

\subsection{Weight Scaling}\label{sec:wtscale}

Given a weighted hypergraph, we want to scale down the weights to make the max-weight small and simultaneously scale down the max-density by a known factor so that we can retrieve the original density value from the scaled one. Since we want to reduce the problem to the unweighted case, we work with the unweighted multi-hypergraph versions (see \Cref{sec:prelim}) of the weighted hypergraphs in question. Thus, the maximum edge-weight would correspond to the max-multiplicity of an edge in the unweighted version. Informally, given a weighted hypergraph $G$ on $n$ vertices, we want to obtain an unweighted multi-hypergraph $H$ such that (a) maximum multiplicity of an egde in $H$ is roughly $O(\log n)$ and (b) given $\rho^*(H)$, we can easily obtain an approximate value of $\rho^*(G)$. We achieve these in \Cref{lem:maxweight,lem:babis} respectively. 

Given any weighted hypergraph $G$, we define $G_q$ as the random hypergraph obtained by independently sampling each hyperedge of $\Gunw$ with probability $q$.

For a parameter $\trho$, define $$q(\trho) := \min\left\{c \eps^{-2}\cdot \frac{\log n}{\trho}, 1\right\}$$ for some large constant $c$ and an error parameter~$\eps>0$.

Our desired multi-hypergraph $H$ will be given by $\Gp$ for some appropriate $\trho$. The following lemma shows that the max-multiplicity of $H$ is indeed small.
\begin{lemma}\label{lem:maxweight}
For $\trho\geq \maxweight(G)/r$, let $H=G_{q(\trho)}$. Then, maximum multiplicity of an edge in $H$ is $O(r\eps^{-2}\log n)$ whp.
\end{lemma} 
\begin{proof}
First note that if $\trho \leq c \eps^{-2}\log n$, then by definition, $H= \Gunw$. Thus, the max multiplicity of an edge in $H$ is $\maxweight(G)\leq r\trho = O(r\eps^{-2}\log n)$. Otherwise, fix an edge $e=\{u_1,\ldots,u_t\}$ in $G$ with weight $w_e$. Thus, the multiplicity of $e$ in $\Gunw$ is $w_e$. Let $X(u_1,\ldots,u_t)$ be the random variable denoting the multiplicity of the hyperedge $\{u_1,\ldots,u_t\}$ in~$H$. We have 
\[\EE[X(u_1,\ldots,u_t)]=w_e\cdot q(\trho) = c \eps^{-2}\cdot \frac{w_e \log n}{\trho} \leq rc\eps^{-2}\log n\,, \]
where the last inequality follows since $w_e\leq \maxweight(G) \leq r\trho$.

Then, using the Chernoff bound (\Cref{fact:chernoff}), we get that $X(u_1,\ldots,u_t) = O(r\eps^{-2}\log n)$  with probability at least $1-1/n^{r+1}$. By a union bound over all $\sum_{t=1}^r \binom{n}{t}\leq O(n^r)$ possible hyperedges, we get that, with probability at least $1-O(1/n)$, we have $X(u_1,\ldots,u_t) = O(r\eps^{-2}\log n)$ for all subsets of nodes $\{u_1,\ldots,u_t\}$ ($t\leq r$) in $H$. Therefore, the max multiplicity of a hyperedge in $H$ is $O(r\eps^{-2}\log n)$ whp.
\end{proof}
At the same time, we also need to ensure that we can retrieve the max-density and a  densest subset of $G$ from that of $H$. The next lemma, which follows directly from Theorem 4 of \cite{mitzenmacher2015scalable}, handles this.

\begin{lemma}\label{lem:babis}
Given a weighted hypergraph $G=(V,E,w)$, let $H=\Gp$ for a parameter $\trho$. Then, following hold simultaneously whp: 

(i) $\forall U\subseteq V:$ $\rho_G(U)\geq (1+\eps)\trho \Rightarrow \rho_{H}(U)\geq c\eps^{-2}\log n$\\
(ii) $\forall U\subseteq V:$ $\rho_G(U)< (1-2\eps)\trho\Rightarrow\rho_{H}(U)< (1-\eps)c\eps^{-2}\log n$
\end{lemma}

It follows from the above lemma that $\rho^*(H) \simeq c\eps^{-2}\log n$ iff $\trho$ is very close to $\rho^*(G)$. We can now make parallel guesses $\trho$ for $\rho^*(G)$ and find the correct one by identifying the guess that gives the desired value of $\rho^*(H)$. We explain this in detail and prove it formally in the next section.

%% file: reductionunweighted.tex
\subsection{Fully Dynamic Algorithm for WDSH using UDSH}\label{sec:wdshp}

We handle the unweighted case UDSH and obtain the following theorem in \Cref{sec:unweighted}. 

\begin{restatable}{theorem}{thmunwtd}\label{thm:unweighted}
Given an unweighted rank-$r$ (multi-)hypergraph $H$ on $n$ vertices and at most $m$ edges with max-multiplicity at least $w^\star$, there exists a fully dynamic data structure \udshp that deterministically maintains a $(1+\eps)$-approximation to the densest subhypergraph problem. The worst-case update time is $O(\max\{(64r\eps^{-2}\log n)/w^\star,1\}\cdot r\eps^{-4}\log^2 n\log m)$ per edge insertion or deletion. The worst-case query times for max-density and densest-subset queries are $O(1)$ and $O(\beta + \log n)$ respectively, where $\beta$ is the output-size.
\end{restatable}

Here, we describe a way to use the above theorem as a subroutine to efficiently solve the dynamic WDSH problem. For the input weighted hypergraph $G$, assume that we know the value of $\maxweight(G)$ and an upper bound $m$ on the number of hyperedges (across all updates) in advance. We can remove these assumptions as follows. We can keep parallel data structures for guesses of $\maxweight(G)$ in powers of $1+\eps$ for the range $[1,W]$ where the weights lie. We extract the solution from the data structure where the guess $w$ is such that the actual current value of $\maxweight(G)$ lies in $[w/2,w]$. This incurs only a factor of $1+\eps$ in the approximation ratio (still giving a near-optimal solution by setting $\eps$ sufficiently small) and a factor of $O(\eps^{-1}\log W)$ in the update time. For the upper bound on $m$, we can simply set it as $R:=\sum_{t=1}^r \binom{n}{t}$, where $r$ is the rank of the hypergraph.  

We observe the following.
\begin{observation}\label{obs:maxwtdensityrel}
In a rank-$r$ weighted hypergraph $G$ with at most $m$ edges, we have $\maxweight(G)/r\leq \rho^*(G)\leq m\maxweight(G)$.
\end{observation}
\begin{proof}
Setting $U$ as the nodes of the heaviest hyperedge in $G$, we have $\rho^*(G)\geq \rho_G(U)\geq \maxweight(G)/r$.

Again, let $U^*$ be the densest subset of vertices in $G$. Then, we have\\ $\rho^*(G)=\rho_G(U^*)=\left(\sum_{e\in E[U^*]}w_e\right)/|U^*|\leq m\maxweight$
\end{proof}

Our algorithm for the dynamic WDSH problem is as follows. 

\mypar{Preprocessing}We keep guesses $\trho_i = (\maxweight/r)(1+\eps)^i$ for $i=0,1,\ldots,\lceil \log_{1+\eps} (rm) \rceil$. Note that by \Cref{obs:maxwtdensityrel}, these are valid guesses for $\rho^*(G)$. For each guess $\trho_i$ and each $j\in \lceil \log_{1+\eps} \maxweight \rceil$, we construct a data structure SAMPLE$(i,j)$ that, when queried, generates independent samples from the probability distribution Bin$\left(\lfloor(1+\eps)^j\rfloor,q(\trho_i)\right)$.\footnote{Bin$(n,p)$ is the Binomial distribution with parameters $n$ and $p$.} Each such data structure can be constructed in $O(\maxweight)$ time so that each query is answered in $O(1)$ time (\cite{BringmannP17}, Theorem 1.2). Parallel to this, for each $i$, we have a copy of the data structure for the UDSH problem, given by $\udshp(i)$. The value of $w^\star$ that we set for $\udshp(i)$ is $\maxweight(G)\cdot q(\trho_i)/2$.  

\mypar{Update processing}On insertion of the edge $e$ with weight $w_e$, for each guess $\trho_i$, query SAMPLE$(i,\lceil \log_{1+\eps} w_e \rceil)$ to get a number $s$, and insert $s$ copies of the unweighted edge $e$ using the data structure $\udshp(i)$. Similarly, on deletion of edge $e$, for each $i$, use $\udshp(i)$ to delete all copies of the edge added during its insertion.

\mypar{Query processing}Denote the value of maximum density returned by $\udshp(i)$ as $\hrho_i$. Let $i^*$ be the largest $i$ such that $\hrho_i\geq (1-\eps)c\eps^{-2}\log n$. On a max-density query for the WDSH problem, we output $\frac{1-2\eps}{1+\eps}\cdot\trho_{i^*}$. For the densest-subset query, we output the densest subset returned by $\udshp(i^*)$. 

\mypar{Correctness}Observe that the hypergraph we feed to $\udshp(i)$ is $G'_{q(\trho_i)}$, where $G'$ is the hypergraph obtained by rounding up each edge weight of $G$ to the nearest power of $(1+\eps)$. Thus, $\rho^*(G)\leq \rho^*(G')\leq (1+\eps)\rho^*(G)$.

For simplicity, we write $G'_{q(\trho_i)}$ as $G'_i$. Note that the value of $w^\star$ provided to each $\udshp(i)$ satisfies the condition in \Cref{thm:unweighted} whp (by the Chernoff bound (\Cref{fact:chernoff})) since the expected value of max-multiplicity of $G'_i$ is $\maxweight(G)\cdot q(\trho_i)$. By \Cref{thm:unweighted}, $\udshp(i)$ returns value $\hrho_i$ such that
\[(1-\eps)\rho^*(G'_i)\leq \hrho_i\leq \rho^*(G'_i)\,. \]
By the definition of $i^*$, we have $\hrho_{i^*}\geq (1-\eps)c\eps^{-2}\log n$. This means $\rho^*(G'_{i^*})\geq (1-\eps)c\eps^{-2}\log n$. Then, by \Cref{lem:babis} (ii), we get $\rho^*(G')\geq (1-2\eps)\trho_{i^*}$. Therefore, we have
\begin{equation}\label{eq:rhoupper}
  \trho_{i^*}\leq \frac{\rho^*(G')}{1-2\eps}\leq \frac{1+\eps}{1-2\eps}\cdot \rho^*(G)\,.  
\end{equation}
Again, note that $\hrho_{i^*+1}< (1-\eps)c\eps^{-2}\log m$. Hence, $\rho^*(G'_{i^*+1})\leq \hrho_{i^*+1}/(1-\eps)<c\eps^{-2}\log m$. Then, by \Cref{lem:babis} (i), it follows that $\rho^*(G')< (1+\eps)\trho_{i^*+1}=(1+\eps)^2\trho_{i^*}$. Hence, we have
\begin{equation}\label{eq:rholower}
  \trho_{i^*}> \frac{\rho^*(G')}{(1+\eps)^2}\geq \frac{\rho^*(G)}{(1+\eps)^2}\,. 
\end{equation}

Thus, from \cref{eq:rhoupper,eq:rholower}, we get
\begin{equation}\label{eq:densapxguarantee}
  \rho^*(G)\geq \frac{1-2\eps}{1+\eps}\cdot\trho_{i^*} \geq \frac{1-2\eps}{(1+\eps)^3}\cdot\rho^*(G)\,.  
\end{equation}

Again, let $U^*$ be the densest subset returned by $\udshp(i^*)$. By \Cref{lem:babis} (ii), we see that \[\rho_{G'}(U^*)\geq (1-2\eps)\trho_{i^*}\geq \frac{1-2\eps}{(1+\eps)^2}\cdot\rho^*(G)\]

Therefore, by the definition of $G'$, we have
\begin{equation}\label{eq:subgapxguarantee}
  \rho^*(G)\geq \rho_G(U^*)\geq \frac{\rho_{G'}(U^*)}{1+\eps}\geq \frac{1-2\eps}{(1+\eps)^3}\cdot\rho^*(G)  
\end{equation}

Given any $0<\delta<1$, we set $\eps=\Theta(\delta)$ small enough so that $\frac{1-2\eps}{(1+\eps)^3}\geq \frac{1}{1+\delta}$.

Therefore, by \cref{eq:densapxguarantee,eq:subgapxguarantee}, the value and the subset that we return on the max-density and densest-subset queries respectively are $(1+\delta)$-approximations to $\rho^*(G)$.

\mypar{Runtime}As noted before, we feed $G'_i$ to $\udshp(i)$. Fix an $i$. Let $\omega_i$ be the max-multiplicity of an edge in $G'_i$. When a hyperedge of $G$ is inserted/deleted, we insert/delete at most $\omega_i$ unweighted copies of that edge to $\udshp(i)$. Therefore, by \Cref{thm:unweighted}, the worst case update time for $\udshp(i)$ is $O(\omega_i\cdot \max\{(64r\eps^{-2}\log n)/w^\star, 1\}\cdot r\eps^{-4}\log^2 n\log m)$. Using the Chernoff bound (\Cref{fact:chernoff}), we have $\omega_i\leq 2\maxweight(G)\cdot q(\trho_i) = 4w^\star$ whp. Also, since $\trho_i\geq \maxweight(G)/r$ for each $i$, we can apply \Cref{lem:maxweight} to get that $\omega_i = O(r\eps^{-2}\log n)$. Hence, the expression simplifies to $O(r\eps^{-2}\log n \cdot r\eps^{-4}\log^2 n\log m ) = O(r^2\eps^{-6}\log^3 n\log m)$. Finally, accounting for all the $O(\log_{1+\eps} rm)=O(\eps^{-1}\log m)$ values of $i$, the total update time is $O(r^2\delta^{-7}\log^3 n\log^2 m)$ (recall that $\delta=\Theta(\eps))$. The max-density query for WDSH is answered by binary-searching on the $O(\eps^{-1}\log m)$ copies of \udshp, which gives a query time of $O(\log \delta^{-1}+\log{\log m})$ by \Cref{thm:unweighted}. Note that the densest-subset query is made only on the relevant copy $i^*$ after we find it, and hence, by \Cref{thm:unweighted}, it takes $O(\beta + \log n)$ time, where $\beta$ is solution-size.

Therefore, we obtain the following theorem that captures our main result.

\begin{theorem}(Formal version of \Cref{thm:main:informal})\label{thmmain}
Given a weighted rank-$r$ hypergraph on $n$ vertices and at most $m$ edges, for any $0<\delta<1$, there exists a randomized fully dynamic algorithm that maintains a $(1+ \delta)$-approximation to the densest subhypergraph problem. The worst-case update time is $O(r^2\delta^{-7}\log^3 n\log^2 m)$ per hyperedge insertion or deletion. The worst-case query times for max-density and densest-subset queries are $O(\log \delta^{-1}+\log\log m)$ and $O(\beta + \log n)$ respectively, where $\beta$ is the output-size. The preprocessing time is $O(\maxweight\delta^{-2}\log m\log \maxweight)$, where $\maxweight$ is the max-weight of a hyperedge.
\end{theorem}

Now all it remains is to solve the unweighted case and prove \Cref{thm:unweighted}. We do this in \Cref{sec:unweighted}.

%% file: udshp.tex
\newif\ifkdd
\kddtrue

\section{Fully Dynamic Algorithm for UDSH}\label{sec:unweighted}

We first design and analyse our algorithm for the dynamic UDSH problem. Then, we design the data structures that actually implement the algorithm efficiently. 

\mypar{Our Algorithm and Analysis}We extend the techniques of \cite{SawlaniW20} for the densest subgraph problem and take the primal-dual approach to solve the UDSH problem. Recall that the input is an unweighted multi-hypergraph $H=(V,E)$ and we want to find the approximate max-density as well as an approximately densest subset of $H$. As is standard, we associate a variable $x_v\in \{0,1\}$ with each vertex $v$ and $y_e\in \{0,1\}$ with each hyperedge $e$ such that $x_v=1$ and $y_e=1$ respectively denote that we include $v$ and $e$ in the solution subset. Relaxing the variables, the primal LP for UDSH (Primal$(H)$) is given below. Following notations similar to \cite{SawlaniW20}, for each vertex $u$ and edge $e$, let $f_e(u)$ and $D$ be the dual variables corresponding to constraints (\ref{const:yexu}) and (\ref{const:primal}) respectively. Then, the dual program Dual$(H)$ is as follows.

\begin{minipage}{0.4\textwidth}
\begin{align}
    &{\bf Primal}(H):\nonumber\\
    &\max \sum_{e\in E} y_e\nonumber\\
    \text{s.t. }&y_e\leq x_{u} \quad\forall u \in e ~\forall e\in E\label{const:yexu}\\
    &\sum_{v\in V} x_v \leq 1\label{const:primal}\\
    &x_v,y_e\geq 0 \quad\forall v\in V, e\in E
\end{align}
\end{minipage}
\vline
\begin{minipage}{0.4\textwidth}
\begin{align}
    &{\bf  Dual}(H):\nonumber\\
    &\min D\nonumber\\
    \text{s.t. }&\sum_{u\in e}f_e(u)\geq 1 \quad \forall e\in E\label{const:loadsum}\\
    &\sum_{e \ni v} f_e(v) \leq D \quad\forall v\in V\\
    &f_e(u)\geq 0 \quad \forall u\in e~\forall e\in E
\end{align}
\end{minipage}

Think of $f_e(u)$ as a ``load'' that edge $e$ puts on node $u$. We can thus interpret Dual$(H)$ as a load balancing problem: each hyperedge needs to distribute a unit load among its vertices such that the maximum total load on a vertex due to all its incident edges is minimized. For each $v\in V$, define $\ell(v):=\sum_{e\ni v} f_e(v)$. Note that if for some feasible solution, some edge $e$ assigns $f_e(v)>0$ to some $v\in e$ and $\ell(v)>\ell(u)$ for some $u\in e\setminus\{v\}$, then we can ``transfer'' some positive load from $f_e(v)$ to $f_e(u)$ while maintaining constraint (\ref{const:loadsum}) and without exceeding the objective value.  Therefore, we can always find an optimal solution to Dual$(H)$ satisfying the following ``local'' property.
\begin{equation}\label{eq:localstability}
\forall e\in E: f_e(v)>0\Rightarrow \ell(v)\leq\ell(u)~\forall u\in e\setminus\{v\}    
\end{equation}
 
We verify that property (\ref{eq:localstability}) is also sufficient to get a {\em global} optimal solution to Dual$(H)$. 

Suppose $\{\hf_e(v): v\in V, e\in E\}$ satisfies property (\ref{eq:localstability}) and also constraint (\ref{const:loadsum}) with equality.\footnote{It is easy to see that in an {\em optimal} solution of Dual$(H)$, constraint (\ref{const:loadsum}) must hold with equality.} Define $\hD := \max_{v} \hl(v)$. Note that $\langle \hf, \hD\rangle$ is a feasible solution to Dual$(H)$ with objective value $\hD$. Since $\rho^*$ is the objective value of an optimal solution to Dual$(H)$, we have $\hD\geq \rho^*$. Now, we show that $\hD\leq \rho^*$. We have
\begin{align*}
 \hD |S|
 &= \sum_{v\in S} \hl(v)\\
 &= \sum_{v\in S}\sum_{e\ni v} \hf_e(v)\\
 &= \sum_{v\in S}\sum_{\substack{e\ni v:\\e\subseteq S}} \hf_e(v) &&\text{since by (\ref{eq:localstability}), $\hf_e(v) = 0$ if $\exists u\in e$ s.t. $u\not\in S$, i.e., $\hl(u)<\hD$}\\
 &= \sum_{e\subseteq S}\sum_{v\in e} \hf_e(v)\\
 &= |E(S)| &&\text{since constraint (\ref{const:loadsum}) holds with equality}
\end{align*}
Therefore, we get $\hD = |E(S)|/|S|\leq \rho^*$, and hence, $\hD=\rho^*$. 

Thus, maintaining the local property (\ref{eq:localstability}) gives us a global optimal solution. We now show in \Cref{thm:approx} that ``approximately'' maintaining property (\ref{eq:localstability}) (see constraint \eqref{const:slacketa}) gives us an approximate global optimal solution.

We define a system of equations Dual$(H,\eta)$ as follows.
\begin{align}
    &\ell(v)=\sum_{e\ni v} f_e(v) &&\forall v\in V\\
    &\sum_{u\in e}f_e(u)= 1 &&\forall e\in E\label{const:exacteqone}\\
    &\ell(v) \leq \ell(u) + \eta &&\forall u\in e\setminus\{v\},~\forall e\in E: f_e(v)>0\label{const:slacketa}\\   
    &f_e(u)\geq 0 &&\forall u\in e~\forall e\in E
\end{align}

\begin{theorem}\label{thm:approx}
Given a feasible solution $\langle \hf,\hl\rangle$ to Dual$(H,\eta)$, we have 
$\rho^*(1-\varepsilon)\leq\hD (1-\varepsilon)< \rho^*$,
where $\hD = \max_{v} \hl(v)$ and $\varepsilon = \sqrt{\frac{8\eta\log n}{\hD}}$. 
\end{theorem}
\begin{proof}
Since $\langle\hf,\hl\rangle$ is a feasible solution to Dual$(H,\eta)$, we see that $\langle \hf, \hD\rangle$ is a feasible solution to Dual$(H)$. Since $\rho^*$ is an optimal solution to Dual$(H)$, we have $\hD\geq \rho^*$  and the left inequality follows. 

Define $S_i := \{v: \hl(v) \geq \hD - \eta i\}$ for $i\geq 0$. For some parameter $0<\gamma<1$, let $k$ be the maximal number such that $|S_i|\geq (1+\gamma)|S_{i-1}|$ for all $i\in [k]$. Thus, $|S_{k+1}|<(1+\gamma)|S_k|$. 

For an edge $e$ incident on $v\in S_k$, consider $u\in e\setminus\{v\}$. We have
 \[u\not\in S_{k+1} \Rightarrow \hl(u)<\hD-\eta(k+1) \leq \hl(v)-\eta \Rightarrow  \hf_e(v) = 0\]
 where the last implication is by (\ref{const:slacketa}). Hence, we get the following.

\begin{observation}\label{obs:etaslack}
For $v\in S_k$, we have $\sum_{e\ni v} \hf_e(v) = \sum_{\substack{e\ni v:\\e\subseteq S_{k+1}}} \hf_e(v)$.
\end{observation}

We try to get a lower bound on $\rho(S_{k+1})$. We see that
\begin{align*}
 (\hD-&\eta k)|S_k|
 \leq \sum_{v\in S_k} \hl(v)
 = \sum_{v\in S_k}\sum_{e\ni v} \hf_e(v) = \sum_{v\in S_k}\sum_{\substack{e\ni v:\\e\subseteq S_{k+1}}} \hf_e(v) \\
 &\leq \sum_{v\in S_{k+1}}\sum_{\substack{e\ni v:\\e\subseteq S_{k+1}}} \hf_e(v)
 = \sum_{e\subseteq S_{k+1}}\sum_{v\in e} \hf_e(v)
 = |E(S_{k+1})|\,. 
\end{align*}
The second equality follows by \Cref{obs:etaslack} and the last one by  (\ref{const:exacteqone}). Therefore, by definition of $k$, we get 
\[\rho(S_{k+1}) = \frac{|E(S_{k+1})|}{|S_{k+1}|}\geq \frac{(\hD-\eta k)|S_k|}{|S_{k+1}|} > \frac{\hD-\eta k}{1+\gamma} > (\hD-\eta k)(1-\gamma)\,.\]

Again, since $|S_k|\geq (1+\gamma)^k|S_0|\geq (1+\gamma)^k$, we have $k\leq \log_{1+\gamma} |S_k| \leq \log_{1+\gamma} n \leq 2\log n/\gamma$. Therefore, we have
\[\rho(S_{k+1})> \left(\hD- \frac{2\eta\log n}{\gamma}\right)(1-\gamma) = \hD\left(1- \frac{2\eta\log n}{\gamma\hD}\right)(1-\gamma)\,.\]

We set $\gamma$ so as to maximize the RHS. Clearly, it is maximized when $\gamma = \frac{2\eta\log n}{\gamma \hD}$, and so, we set $\gamma := \sqrt{\frac{2\eta\log n}{\hD}}$. Hence, we get
\begin{equation}\label{eq:approxeq}
\rho^* \geq \rho(S_{k+1}) > \hD\left(1- \gamma\right)^2> \hD(1-2\gamma)=\hD\left(1- \sqrt{\frac{8\eta\log n}{\hD}}\right)\,,    
\end{equation}
which completes the proof.
\end{proof}
By \Cref{thm:approx}, we see that if we can find $\hD$, i.e., a feasible solution to Dual$(H,\eta)$, then we can get a $(1+\eps)$-approximation to $\rho^*$, where $\eps=\sqrt{\frac{32\eta\log n}{\hD}}$. This means that given $\eps$, we initially need to set $\eta=\frac{\eps^2\hD}{32\log n}$. But we do not know the value of $\hD$ initially, and in fact, that's what we are looking for. However, we shall initially have an estimate $\tD$ of $\hD$ such that $\tD\leq \hD\leq 2\tD$. We set $\eta:=\frac{\eps^2\tD}{32\log n}$.
Since $\tD\leq \hD$, we get $\sqrt{\frac{8\eta\log n}{\hD}}\leq \frac{\eps}{2}$, which, by \Cref{thm:approx} and \cref{eq:approxeq}, implies a $(1+\eps)$-approximation to UDSH. Thus, the next corollary follows.

\begin{corollary}\label{cor:correctness}
Given $0<\eps<1$ and $\tD\leq \hD$, setting $\eta=\frac{\eps^2\tD}{32\log n}$ ensures that
\begin{itemize}
\setlength\itemsep{0em}
    \item $S_{k+1}$ is a $(1+\eps)$-approximation to the densest subset
    \item $\hD(1-\frac{\eps}{2})$ is a $(1+\eps)$-approximation to the maximum density
\end{itemize}
where $\hD$ and $S_{k+1}$ are defined as above.
\end{corollary}

Thus, we focus on finding a feasible solution to Dual$(H,\eta)$, where $\eta=\frac{\eps^2\tD}{32\log n}$ for a given estimate $\tD$ satisfying $\tD\leq \hD\leq 2\tD$. Note that if we have $\eta\geq 1$, then we can maintain constraint (\ref{const:slacketa}) with some positive slack while having integer loads on the vertices. This means that we are allowed to simply assign the unit load of an edge $e$ entirely on some vertex $u\in e$. 
Assume that we know a lower bound $w^\star$ on the max-multiplicity of a hyperedge in the graph. If $w^\star\geq 64r\eps^{-2}\log n$, then it already implies that $\eta\geq 1$ since $\hD\geq  \rho^* \geq 64\eps^{-2}\log n$ and hence, $\tD\geq \hD/2 \geq 32\eps^{-2}\log n$. Otherwise, we duplicate each hyperedge $\lceil(64r\eps^{-2} \log n)/w^\star\rceil$ times, so that we are ensured that $\rho^*\geq 64\eps^{-2}\log n$, implying $\eta\geq 1$ as before. Once we have $\eta\geq 1$ and are allowed to assign the entire load of an edge on a single node in it, our problem reduces to the following hypergraph ``orientation'' problem. 

{\bf Problem} (Hypergraph Orientation). Given an unweighted multi-hypergraph $H=(V,E)$ and a parameter $\eta\geq 1$, for each edge $e\in E$, assign a vertex $v\in e$ as its head $h(e)$, such that
\begin{equation}\label{eq:indegeta}
   \forall e\in E: h(e)=v \Rightarrow \din(v) \leq \din(u) + \eta~\forall u\in e\setminus\{v\}
\end{equation}
where $\din(v):= |\{e\in E: h(e) = v\}|$.

Call the operation that reassigns the head of an edge to a different vertex as an {\em edge rotation}. The next lemma (proof deferred to the supplement) shows that we can design a dynamic algorithm that solves the hypergraph orientation problem, thus getting a $(1+\eps)$-approximation to UDSH, by making a small number of hyperedge rotations per insertion/deletion. 
\begin{lemma}\label{lem:recourse}
There exists a dynamic algorithm that solves the hypergraph orientation problem by making $\hD/\eta$ hyperedge rotations per insertion/deletion, where $\hD:= \max_{v\in V} \din(v)$.
\end{lemma}
\begin{proof}
Call a hyperedge $e$ tight if $h(e)=v$ and $\exists u \in e$ such that $\din(v) = \din(u) +\eta$. Suppose that we have maintained (\ref{eq:indegeta}) until an edge $e$ is inserted. We assign the vertex $v$ with minimum $\din$ as its head $h(e)$. Suppose this violates (\ref{eq:indegeta}) for some edge $e'$. Therefore, before the insertion of $e$, the hyperedge $e'$ was tight with $h(e')=v$. We do rotate$(e')$, which might violate (\ref{eq:indegeta}) for some edge $e''$ that was tight before. We then perform rotate$(e'')$ and continue this to obtain a chain of tight hyperedges going ``outwards'' from $v$. Note that the chain terminates because $\din$ of the heads strictly decreases along the chain. Thus, rotating the last hyperedge in this chain doesn't cause a violation of (\ref{eq:indegeta})  for any edge, and hence, we maintain (\ref{eq:indegeta}). 

Note that the number of rotations we perform is same as the size of this maximal chain of tight hyperedges, which is at most $\hD/\eta = O(\eps^{-2}\log n)$ since $\din(v)\leq \hD$ and $\din(h(e_2))\leq \din(h(e_1))-\eta$ for two successive hyperedges $e_1$ and $e_2$ in the chain. Finally, recall that we duplicate each edge $O(r\eps^{-2}\log n)$ times when it is inserted, and hence, we perform at most $O(r\eps^{-2}\log n \cdot \eps^{-2}\log n) = O(r\eps^{-4}\log^2 n)$ rotations in total per hyperedge-insertion. 
The case of deletion is handled similarly, except that the maximal chain of tight hyperedges is oriented ``inwards'' to $v$.      
\end{proof}
Hence, for $\eta=\frac{\eps^2\tD}{32\log n}$, where $\hD\leq 2\tD$, we make at most $O(\eps^{-2}\log n)$ edge rotations. Next, we design a data structure that actually performs these rotations and solves the hypergraph orientation problem with an update time that incurs only a small factor over the number of edge rotations.

\mypar{Implementation details}Given a parameter $\tD$, we construct a data structure HOP($\tD$) that maintains the ``oriented'' hypergraph satisfying (\ref{eq:indegeta}). We describe it in detail in Data Structure \ref{ds:hop}. 

The following lemmas give the correctness and runtime guarantees of the data structure.
\begin{lemma}\label{lem:lazystabmaintain}
After each insertion/deletion, the data structure HOP$(\tD)$ maintains constraint (\ref{eq:indegeta}).
\end{lemma}
\begin{proof}
 Similar to \Cref{lem:recourse}, we can show that due to the chain of rotations, each insertion/deletion leads to only a single increment/decrement. Let $h(e)=v$ and $u$ be any vertex in $e\setminus \{v\}$. The constraint $\din(v)\leq \din(u) + \eta$ can be violated either due to a decrement to $\din(u)$ or an increment to $\din(v)$. 
 
 Consider the case when we have a decrement to $\din(u)$. Observe that this happens only when $u$ cannot find a tight ``out-edge''. Hence, $\din^{(u)}(v) < \din(u) + \eta/2$ at that time. Now, note that each vertex $w$ informs its $\din$ value to all vertices in the hyperedges in In$(w)$ every $\eta/4$ iterations. Hence, there are at most $\eta/4$ updates to $\din(v)$ between two successive iterations where it informs its $\din$ value to $u$. Therefore, at any point, $\din(v)\leq \din^{(u)}(v) + \eta/4$, i.e., at the time $\din(u)$ is decremented, we have $\din(v)\leq \din(u) + 3\eta/4$, implying that the constraint is maintained.    
 
 Now we turn to the other case when we have an increment to $\din(v)$. This happens only when $v$ cannot find a tight ``in-edge'' among the $4\din(v)/\eta$ edges it probes. Consider the last $\eta/4$ updates where $\din(v)$ was incremented. There exists an update among these where the edge $e$ was probed. Hence, at that time, we had $\din(v)< \din(u)+\eta/2$. If there has been no decrement to $\din(u)$ in the following updates, then, since there has been at most $\eta/4$ successive increments to $\din(v)$, we must have $\din(v)< \din(u)+3\eta/4$, implying that the constraint is maintained. Otherwise, consider the last update where $\din(u)$ was decremented. As noted before, we had $\din(v)\leq \din(u) + 3\eta/4$ at that time. There has been at most $\eta/4$ increments to $\din(v)$ since then, and hence, we have $\din(v)\leq \din(u) + \eta$. Thus, in any case, the constraint  (\ref{eq:indegeta}) is maintained after each insertion/deletion.  
\end{proof}

\begin{lemma}\label{lem:datastrcorrectness}
If $\tD\leq \hD$, then the operations querysubset and querydensity of HOP$(\tD)$ return a $(1+\eps)$-approximation to the densest-subset and max-density queries respectively.
\end{lemma}
 
\begin{proof}
Follows directly from \Cref{lem:lazystabmaintain} and \Cref{cor:correctness}
\end{proof}


\begin{lemma}\label{lem:updatetime}
If $\hD\leq 2\tD$, then the operations insert and delete of HOP$(\tD)$ take $O(r\eps^{-4}\log^2 n)$ and $O(r\eps^{-2}\log n))$ time respectively. The operation querydensity takes $O(1)$ time and querysubset takes $O(\beta+\log n)$ time, where $\beta$ is the solution-size.
\end{lemma}

\begin{proof}
An insert operation first takes $O(r)$ time to find the vertex in $e$ with minimum $\din$. Adding the label $e$ to In$(v)$, where $v=h(e)$, and to Out$(u)$ for each $u\in e\setminus\{v\}$ takes $O(r)$ time. Next, each iteration of the while loop finds a tight edge using tightinedge, locates its minimum $\din$ vertex, and performs a rotate operation. The tightinedge operation probes at most $4\din(v)/\eta=O(\hD/\eta)$ edges, finding a min $\din$ vertex in each of these edges in $O(r)$ time, and a rotate operation takes $O(r)$ time. Hence, each iteration of the while loop takes $O(r\hD/\eta)$ time. Observe that the number of iterations of the while loop is the size of a maximal chain of tight edges which is at most $O(\hD/\eta)$ by similar argument as in the proof of \Cref{lem:recourse}. Therefore, the while loop takes time $O(r(\hD/\eta)^2)$. Finally, it performs a single increment operation, which probes $4\din(v)/\eta=O(\hD/\eta)$ edges and updates the $\din^{(u)}(v)$ values for each node $u$ in those edges. Thus, it takes $O(r\hD/\eta)$ time. Therefore, the total time an insert operation takes is dominated by the while loop, which is $O(r(\hD/\eta)^2)=O(r\eps^{-4}\log^2 n)$. 

The analysis of the delete operation is similar, except that a call of the tightoutedge operation takes $O(1)$ time because we obtain Out$(v)$.max from a max-priority queue, implying that each iteration of the while loop takes $O(r)$ time (due to the rotate operation). Thus, the while loop takes $O(r\hD/\eta)=O(r\eps^{-2}\log n)$ time in total, same as a single decrement operation. Therefore, a delete operation takes $O(r\eps^{-2}\log n)$ time. 

The operation querydensity just extracts the max element from the BST Indegrees, taking $O(1)$ time. The operation querysubgraph calls densestsubg$(\gamma)$ that keeps on growing the solution set by traversing the BST Indegrees and reads $\beta$ elements in total where $\beta$ is the solution. The time taken is $O(\beta+\log n)$. 
\end{proof}

\input{datastrcs}


\mypar{Completing the Algorithm}We built the data structure HOP$(\tD)$ and showed that we obtain the claimed bounds on the approximation factor and the update time as long as $\tD\leq \hD\leq 2\tD$. Now, to get access to such an estimate, we keep parallel data structures HOP$(\tD)$ for $O(\log m)$ guesses of $\tD$ in powers of $2$. Then, we  show that we can maintain an ``active'' copy of HOP corresponding to the correct guess, from which the solution is extracted. Thus, we incur only an $O(\log m)$ overhead on the total update time for an edge insertion/deletion. This part is very similar to Algorithm 3 of  \cite{SawlaniW20} and we sketch it below.

The data structure UDSHP keeps $O(\log m)$ copies of the data structure HOP: for $i\in [\lceil\log m\rceil]$, we guess that $\tD = 2^{i-1}$. The copy where we have $\tD\leq \hD\leq 2\tD$ is called the \emph{active} copy. The max density and densest subgraph queries are always answered by querying this copy. The final algorithm is given as follows. 

We initialize the index of the active copy as $0$. On an edge insertion, we duplicate and insert it\\ $\max\{(64r\eps^{-2}\log n)/w^*,1\}$ times as mentioned before. We insert each edge to all copies of the data structure whose index is greater than the active index. Note that we can afford to do it for these copies since the length of the maximal chain is bounded by $\hD/\eta(\tD)\leq 2\tD/\eta(\tD)\leq O(\eps^{-2}\log n)$ as desired. For the active copy $i$, we query its max density and then insert the edge only if it turns out to be smaller than $2^i$. Otherwise, we increment the active index by $1$. Now, for the copies where we haven't inserted the edge yet, we first query the $\din$ of the vertices contained in the edge, and if we see that the minimum of these values is less than the threshold of $2\tD$ for that copy, we insert the edge. Otherwise we keep it in a {\em pending} list corresponding to that copy. 

On an edge deletion, we delete all copies of the edge we had added during its insertion from all copies whose index is greater than the active index. For the active copy, we query its max density and then delete the edge only if it's larger than $2^{i-1}$. Otherwise, we decrement the active index by $1$. Now, for the remaining copies, we first check if this edge is present in the pending list for that copy, and if so, we delete one copy of the edge from that list. Otherwise, we check if there's any edge in the pending list that contains the head of this edge, and if so, after deleting the former edge from the graph, we insert the pending edge. This ensures that if a copy becomes active after some deletions, the pending edges are inserted and the copy processes the entire graph. Therefore, we always maintain the active index correctly. We answer the max density and densest subgraph queries from the active copy only.

\thmunwtd*

\begin{proof}
 The correctness follows from \Cref{lem:datastrcorrectness} and the discussion above.
 
For the runtime, observe that for each edge insertion, UDSHP runs the HOP.insert operation\\ $\max\{(18r\eps^{-2}\log n)/w^*,1\}$ times by duplicating it that many times. We do this for each of the $O(\log m)$ copies where we insert the edge. By \Cref{lem:updatetime}, each HOP.insert operation takes  $O(r\eps^{-4}\log^2 n)$ time. Hence, in total, an edge insertion takes $O(\max\{(r\eps^{-2}\log n)/w^*,1\}\cdot r\eps^{-4}\log^2 n\log m)$ time. Similarly, an edge deletion calls the HOP.delete operation $O(\max\{(18r\eps^{-2}\log n)/w^*,1\}\cdot\log m)$ times and with each such HOP.delete operation, there might be an HOP.insert operation. By \Cref{lem:updatetime}, we see that the runtime of a HOP.insert operation dominates that of a HOP.delete operation, and hence, the total update time for a deletion is $O(\max\{(18r\eps^{-2}\log n)/w^*,1\}\cdot r\eps^{-4}\log^2 n\log m)$. The query times follow immediately from \Cref{lem:updatetime} since we query only the active copy.      
\end{proof}

%% file: datastrcs.tex
\newgeometry{margin=0.5in}
\captionof{Data Structure}{The algorithms for HOP$(\tD)$ that solves the hypergraph orientation problem}\label{ds:hop}\fbox{
\parbox{\linewidth}{
\begin{multicols}{2}
[{\bf Input:} Unweighted hypergraph $H=(V,E)$, parameters $\eps, \tD$
]
$n\leftarrow |V|$
\vspace{2mm}

$\eta\leftarrow \frac{\eps^2\tD}{32\log n}$
\vspace{2mm}

Indegrees $\leftarrow$ Balanced binary search tree for $\{\din(v): v\in V\}$
\vspace{2mm}

Each hyperedge $e$ maintains a list of vertices that it contains and has a pointer $h(e)$ to the head vertex.

Each vertex $v$ maintains the following data structures:
\begin{itemize}
    \item $\din(v)$: Number of hyperedges $e$ such that $h(e) = v$
    \item In$(v)$: List of hyperedges (labels only) where $v$ is the head
    \item Out$(v)$: Max-priority queue of $\{e\in E: h(e)\neq v\}$, indexed by \hspace*{1cm}$\din^{(v)}(h(e))$ 
      \item $\din^{(v)}(u)$ ($\forall e \in$ Out$(v)~\forall u\in e$): $\din(u)$ from $v$'s perspective
\end{itemize}
\end{multicols}
\begin{multicols}{3}

\begin{algorithm}[H]
\begin{algorithmic}[1]
    \STATE $z\leftarrow h(e)$
    \STATE remove $e$ from In$(z)$
    \STATE remove $e$ from Out$(u)$ for each $u\in e\setminus \{z\}$
    \STATE $h(e)\leftarrow v$; \quad add $e$ to In$(v)$
    \STATE add $e$ to Out$(u)$ for each $u\in e\setminus \{v\}$
\end{algorithmic}
\caption{{\sc rotate}$(e,v)$\label{alg:rotate}}
\end{algorithm}

\begin{algorithm}[H]
\begin{algorithmic}[1]
\FOR{$e\in$ \{next $4\din(v)/\eta$ edges in In$(v)$\}}
\STATE $u\leftarrow \text{arg}\min_{z\in e} \din(z)$
\IF{$\din(u)\leq \din(v)-\eta/2$} 
\STATE return $e$
\ENDIF
\ENDFOR
\STATE return null
\vspace{5mm}
\end{algorithmic}
\caption{{\sc tightinedge}$(v)$\label{alg:tightinedge}}
\end{algorithm}

\begin{algorithm}[H]
\begin{algorithmic}[1]
\STATE $v\leftarrow \text{arg}\min_{z\in e} \din(z)$
\STATE $h(e)\leftarrow v$
\STATE add $e$ to In$(v)$
\STATE add $e$ to Out$(u)$ for each $u\in e\setminus \{v\}$
\WHILE{tightinedge$(v) \neq$ null}
\STATE $f \leftarrow$ tightinedge$(v)$
\STATE $v\leftarrow \text{arg}\min_{z\in f} \din(z)$
\STATE rotate$(f,v)$
\ENDWHILE
\STATE increment$(v)$
\vspace{2mm}
\end{algorithmic}
\caption{{\sc insert}$(e)$\label{alg:hopinsert}}
\end{algorithm}

\begin{algorithm}[H]
\begin{algorithmic}[1]
\STATE $\din(v)\leftarrow \din(v) + 1$\\
\STATE Update $\din(v)$ in Indegrees 
\FOR{$e \in$ \{next $4\din(v)/\eta$ edges in 
In$(v)$\}}
\FOR{$u\in e$}
\STATE $\din^{(u)}(v) \leftarrow \din(v)$
\ENDFOR
\ENDFOR
\end{algorithmic}
\caption{{\sc increment}$(v)$\label{alg:increment}}
\end{algorithm}

\begin{algorithm}[H]
\begin{algorithmic}[1]
\STATE $e \leftarrow$ Out$(v)$.max
\IF{$\din^{(v)}(h(e))\geq \din(v)+\eta/2$} 
\STATE return $e$
\ENDIF
\STATE return null
\end{algorithmic}
\vspace{0.9cm}
\caption{{\sc tightoutedge}$(v)$\label{alg:tightoutedge}}
\end{algorithm}

\begin{algorithm}[H]
\begin{algorithmic}[1]
\STATE $v\leftarrow h(e)$
\STATE remove $e$ from In$(v)$
\STATE remove $e$ from Out$(u)$ for each $u\in e\setminus \{v\}$
\WHILE{tightoutedge$(v) \neq$ null}
\STATE $f \leftarrow$ tightoutedge$(v)$
\STATE $w\leftarrow h(f)$
\STATE rotate$(f,v)$; \quad $v\leftarrow z$
\ENDWHILE
\STATE decrement$(v)$
\vspace{7mm}
\end{algorithmic}
\caption{{\sc delete}$(e)$\label{alg:hopdelete}}
\end{algorithm}

\begin{algorithm}[H]
\begin{algorithmic}[1]
\STATE $\din(v)\leftarrow \din(v) - 1$
\STATE Update $\din(v)$ in Indegrees
\FOR{$e \in$ \{next $4\din(v)/\eta$ edges in In$(v)$\}}
\FOR{$u\in e$} 
\STATE $\din^{(u)}(v) \leftarrow \din(v)$
\ENDFOR
\ENDFOR
\end{algorithmic}
\caption{{\sc decrement}$(v)$\label{alg:decrement}}
\end{algorithm}

\begin{algorithm}[H]
\begin{algorithmic}[1]
\STATE $\hD\leftarrow$ Indegrees.max
\STATE $A \leftarrow \{v: \din(v) \geq \hD\}$
\STATE $B \leftarrow \{v: \din(v) \geq \hD - \eta \}$
\WHILE{$|B|/|A|\geq 1+\gamma$}
\STATE $\hD\leftarrow \hD-\eta$; \quad $A \leftarrow B$
\STATE $B \leftarrow \{v: \din(v) \geq \hD - \eta\}$
\ENDWHILE
\STATE return $B$
\end{algorithmic}
\vspace{4mm}
\caption{{\sc densestsubset}$(\gamma)$\label{alg:densestsub}}
\end{algorithm}

\begin{algorithm}[H]
\begin{algorithmic}[1]
\STATE $\hD \leftarrow$ Indegrees.max
\STATE $\gamma \leftarrow \sqrt{2\eta \log n/\hD}$
\STATE return densestsubset$(\gamma)$
\end{algorithmic}
\caption{{\sc querysubset}$()$\label{alg:querysubset}}
\end{algorithm}

\begin{algorithm}[H]
\begin{algorithmic}[1]
\STATE return (Indegrees.max)$\cdot(1-\frac{\eps}{2})$
\end{algorithmic}
\caption{{\sc querydensity}$()$\label{alg:querydensity}}
\end{algorithm}
\end{multicols}
}
}
\newgeometry{margin=1in}
\setlength{\parskip}{0.5em}

%% file: experiments.tex
\section{Experiments}
\label{sec:exp}
In this section, we present extensive experimental evaluations of our algorithm. We consider weighted and unweighted hypergraphs in both insertion-only and fully dynamic settings, leading to a total of four combinations. We call our algorithms \udshp and \wdshp for the unweighted and weighted settings respectively. For each combination, we compare the accuracy and efficiency of \udshp or \wdshp to that of the baseline algorithms.
Furthermore, we study the trade-off between accuracy and efficiency for \udshp and \wdshp in various settings.

\mypar{Datasets}
For evaluating \udshp, 
we collect real-world temporal hypergraph datasets, as we describe below. See~\Cref{tab:dataset_desc} for a summary of these hypergraphs.

\emph{Publication datasets.} We consider two publication datasets:
DBLP~\cite{dblp} and Microsoft Academic Graph (MAG) with the geology tag~\cite{Sinha-2015-MAG,Benson-2018-simplicial}\footnote{\label{source1}Source:~\url{https://www.cs.cornell.edu/~arb/data/}}. We encode each author as a vertex and each publication as a hyperedge with the publication year serving as the timestamp. In the fully dynamic case, we maintain a sliding window of $10$ years, by removing hyperedges that are older than $10$ years. We treat multiple papers by the same set of authors as a single hyperedge  and report densest subgraph at the end of each year.

\emph{Tag datasets.} We consider 3 tag datasets: math exchange~\cite{Benson-2018-simplicial}\cref{source1}, stack-overflow~\cite{Benson-2018-simplicial}\cref{source1}, and ask-ubuntu~\cite{Benson-2018-simplicial}\cref{source1}. In each of these datasets, a vertex corresponds to a tag, and a hyperedge corresponds to a set of tags on a post or a question in the respective website. In the fully dynamic model, we maintain a sliding window of 3 months. In both insertion only and dynamic settings, we report the densest subgraph at an interval of 3 months. 

\emph{Drug Abuse Warning Network(DAWN) dataset.} This dataset\cref{source1} is generated from the national health surveillance system that records drug abuse related events leading to an emergency hospital visit across  USA~\cite{Benson-2018-simplicial}. We construct a hypergraph where the vertices are the drugs and
a hyperedge corresponds to a combination of drugs taken together at the time of abuse. We maintain the most recent 3 months records in fully dynamic and insertion only settings. The reporting of the maximum density is done at an interval of 3 months.

\mypar{Weighted Datasets} Each of the datasets described above are unweighted. We are not aware of any publicly available weighted temporal hyperagphs. For our weighted settings, we transform the unweighted temporal hypergraph into a weighted temporal hypergraph by the following process. For each edge, we assign it an integer weight sampled uniformly at random from $[1,100]$.

\begin{table}[h]
\centering
  \caption{Description of our dataset with the key parameters, \#vertices($n$), \#hyperedges($m$), maximum size (\#hyperedges) in the dynamic setting ($m_{\Delta}$), \#rank($r$). 
}
  \begin{tabular}{ccccc}
    \toprule
    Dataset  &$n$ & $m$ &  $m_{\Delta}$ & $r$ \\
    \midrule
    dblp-all & 2.56M & 3.16M & 1.99M & 449   \\
    tag-math-sx & 1.6K & 558K & 21.3K & 5 \\
    tag-ask-ubuntu & 3K & 219K & 10.4K & 5 \\
    tag-stack-overflow & 50K & 12.7M & 50.6K & 5 \\
    dawn & 2.5K & 834K & 11.5K & 16 \\
    coauth-MAG-geology & 1.25M & 960K & 216.3K & 25 \\
    \bottomrule
  \end{tabular}
  \label{tab:dataset_desc}
\end{table}

\mypar{Implementation Details}
The implementation details of our algorithm are given in Data Structure \ref{ds:hop}.\footnote{Our code is available (anonymously): \href{https://drive.google.com/drive/folders/1JOepPhEaHP8nlbf8wm0D7J5gw7o_ydMh?usp=sharing}{Link to Code Repo}} In implementing \Cref{alg:densestsub}, we consider all potential 
subsets $B$ by ignoring the  condition on line 3, and report the subset with the largest density among these choices. We implement all algorithms in C++ and run our experiments on a workstation with  256 GB DDR4 DRAM memory and Intel Xeon(R) 5120 CPU 2.20GHz processor running Ubuntu 20 operating system.

\mypar{Baseline Algorithms}
We consider two main baselines algorithms. 
\begin{inparaenum}[\bfseries (1)]
\item The first one is an exact algorithm, denoted as \exact, that
computes the exact value of the densest subhypergraph at every reporting interval of the dataset. We use google OR-Tools to implement an LP based solver for the densest subhypergraph~\cite{hu2017maintaining,ortools} in all the four settings we consider. 
\item The second one is the dynamic algorithm 
for maintaining densest subhypergraph by Hu~\etal~\cite{hu2017maintaining}; we call it \hwc.
It takes $\eps_H$ as an input accuracy parameter and produces a $(1+\eps_H)r$ and $(1+\eps_H)r^2$-approximate densest subhypergraph in the insertion only and fully dynamic model, respectively. For the weighted hypergraphs we modify the \hwc implementation -- each edge with weight $w_e$ is processed by creating $w_e$ many copies of that edge.
\end{inparaenum}

\mypar{Parameter Settings} 
Both \hwc and our algorithms \udshp and \wdshp take an accuracy parameter $\eps$ as part of the input. 
However, it is important to note that the accuracy parameter $\eps$ for both the algorithms are not directly comparable. 
\udshp or \wdshp guarantees to  maintain an $(1+ \eps)$-approximate solution in both insertion only and fully dynamic settings, where as \hwc maintains  $(1+\eps)r$ and $(1+\eps)r^2$ approximate solutions for insertion and fully dynamic settings respectively. 
Thus, for a  fair comparison between the algorithms, we run \udshp (or \wdshp) and \hwc with different values of $\eps$ such that their accuracy is comparable. 
We use $\eps_H$ to denote the parameter for \hwc to make this distinction clear. 
For various settings and datasets, we use different parameters and specify them in the corresponding plots. 
We emphasis here that the motivation behind the choices of the parameters is to compare the update time of \udshp and \wdshp to that of \hwc while in the {\em small} approximation error regime.

\mypar{Accuracy and Efficiency Metrics}  To measure the accuracy of \udshp, \wdshp, and \hwc, we use \emph{relative error percentage} with respect to \exact.  It is defined as $\frac{|\rho(\textsc{Alg}, t) - \rho(\exact, t)|}{\rho(\exact, t)} \times 100\%$, where $\rho(X, t)$ is the density estimate by algorithm $X$ at time interval~$t$. We also compute the average relative error of an algorithm by taking the average of the relative errors over all the reporting intervals. Measuring efficiency is quite straightforward: we compare the average wall-clock time taken over the operation during each reporting interval. We also use the overall average time (taken over all the reporting intervals) as an efficiency comparison tool.

\input{fully-dynamic}

\input{insertion-only-case}


%% file: fully-dynamic.tex
\subsection{Fully Dynamic Case}
\label{subsec:fully}
In this subsection, we consider the fully dynamic settings where the hyperedges can be both inserted and deleted. 
We perform experiments for hypergraphs with unweighted as well as weighted edges.
For both the cases, we first compare the accuracy and the efficiency of our algorithm against the baselines. And then we analyze the accuracy vs efficiency trade-off of \udshp and \wdshp.

\mypar{\textsc{Unweighted Hypergraphs}} We first discuss our findings for unweighted hypergraphs.
\begin{figure}[H]
    \centering
    \includegraphics[width=\textwidth]{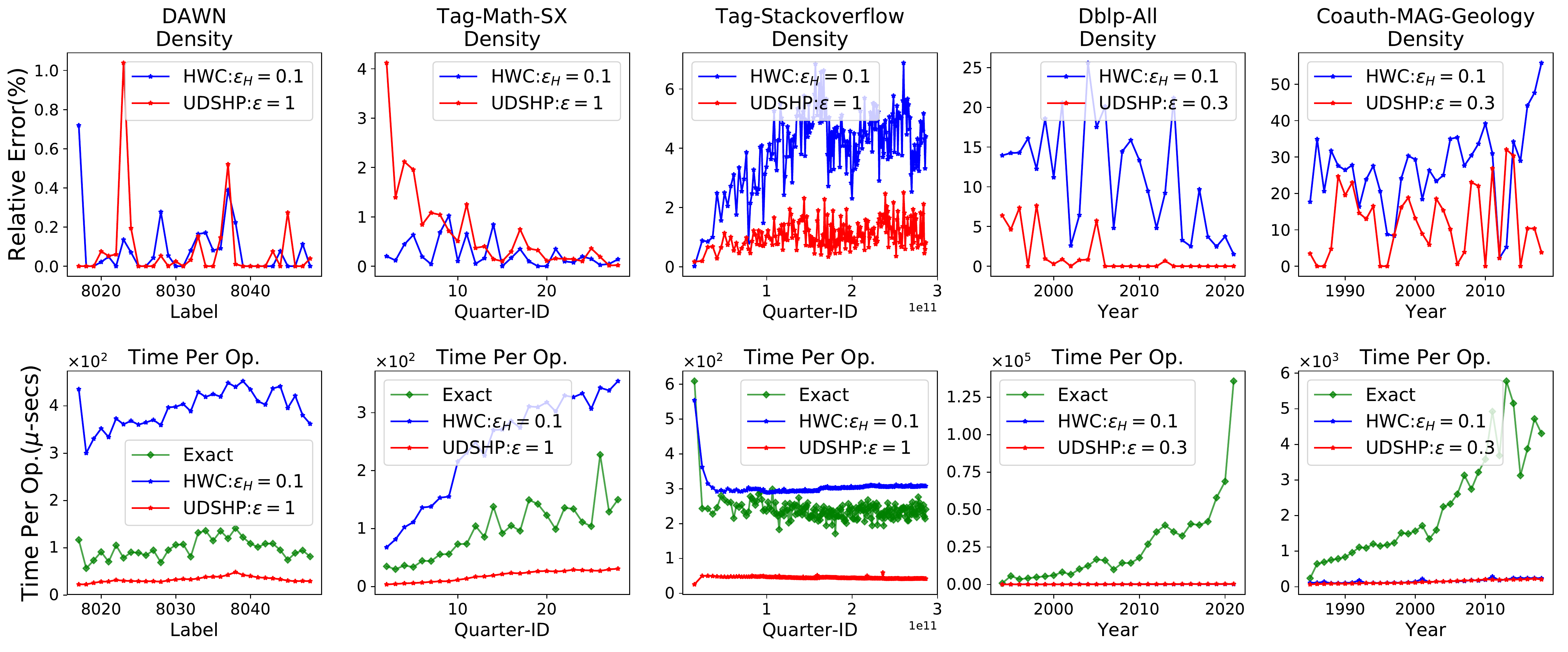}
    \caption{Accuracy and Efficiency Comparison for Unweighted Dynamic Hypergraphs: The top row shows the relative error in the reported maximum density by \udshp and \hwc with respect to \exact when run with the specified parameters. The bottom row plots the average update time taken by \udshp, \hwc, and \exact for each reporting intervals. For each dataset (column), the parameter settings are identical.}
    \label{fig:dyn_unweighted_error_time}
\end{figure}

\begin{figure}[H]
\centering
\includegraphics[scale=0.3]{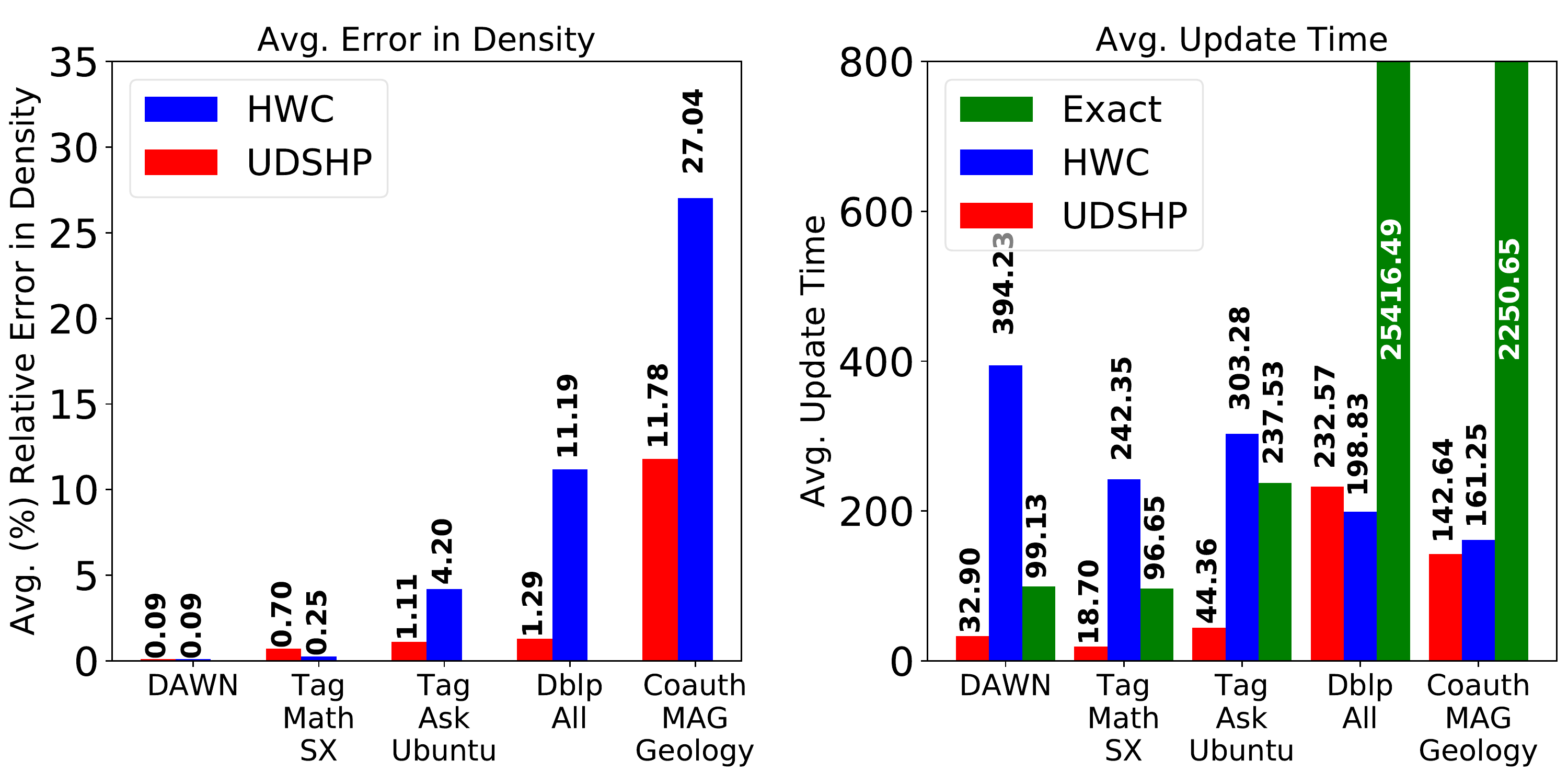}
\caption{Average Accuracy and Efficiency Comparison for fully dynamic unweighted hypergraphs: on the left, we plot the average relative error of \udshp and \hwc over all the reporting intervals for each dataset, and on the right, we compare the average update time of \udshp, \hwc, and \exact over the entire duration.}
\label{fig:avg_time_error_dyn_unweighted}
\end{figure}

\mypar{Accuracy and Efficiency Comparison against Baselines}
In~\Cref{fig:dyn_unweighted_error_time}, we compare the accuracy and the efficiency of \udshp against the baselines for the unweighted hypergraphs. 
In the top row, we first compare the accuracy of \udshp and \hwc in terms of relative error percentage with respect to \exact. In the bottom row, we plot the average time taken per operation by \exact, \udshp, and \hwc during each reporting interval. For each dataset, the parameters are identical for the top row and bottom row plots. We reiterate that the input parameters for \udshp and \hwc are chosen to compare \udshp and \hwc in the low relative error regime. We highlight our main findings below.

We observe that for smaller hypergraphs (\dawn, \tagm), \udshp and \hwc achieve impressive accuracy, however \udshp is consistently more than 10x faster than \hwc. In fact, \hwc is several times slower compared to \exact. On the other hand, \udshp is 3x-5x times faster compared to \exact. 
As the sizes of the hypergraphs increase, \exact gets much slower compared to \udshp and \hwc as LP solvers are known to have scaling issues.\footnote{Note that, although the {\tagsf} hypergraph has overall more edges than the \cmag hypergraph, the reporting interval for the latter is much longer than the former. Thus at any given interval, the latter hypergraph contains more edges leading to larger update times.}
For larger datasets, \udshp maintains a clear edge in terms of accuracy over \hwc even when their update times are almost identical or better for \udshp, as demonstrated by the last three columns. In~\Cref{sec:addlexp}, we show results for some more parameter settings (\Cref{fig:dyn_unweighted_error_time_2_eps}).  

To further quantify the gain of \udshp, in~\Cref{fig:avg_time_error_dyn_unweighted}, we compare the performances of \udshp against \hwc and \exact in terms of average relative error and average update time, where the average is taken over all the reporting intervals. We make several interesting observations. (1)
\udshp is 3x-5x faster than \exact for small hypergraphs; the gain becomes massive 10x-15x for larger graphs.
(2) Compared against \hwc, the avg. update time for \udshp can be 10x-12x smaller (\dawn and \tagm) while maintaining almost the same average relative error of less than 1\%. (3) At the other end of the spectrum, for almost the same average update time, \udshp can offer 55\%-90\% improvement in accuracy over \hwc (\cmag and \dblpall). (4) \hwc performs worse than \exact for smaller datasets, being slower by 3x-5x factors (\dawn and \tagm). 

\mypar{Efficiency vs Accuracy trade offs for UDSHP}In~\Cref{fig:dyn_unweighted_time_max_avg_density}, we plot {\em average update time}, and {\em average and max relative error} for \udshp for different values of $\epsilon\in \{1.0, 0.7, 0.5\}$. The max relative error is the maximum of the relative error over all the reporting intervals. 
As expected, when $\epsilon$ decreases, the update time increases and the average and maximum relative error incurred by \udshp decreases. 
\begin{figure}[H]
    \centering
    \includegraphics[scale=0.3]{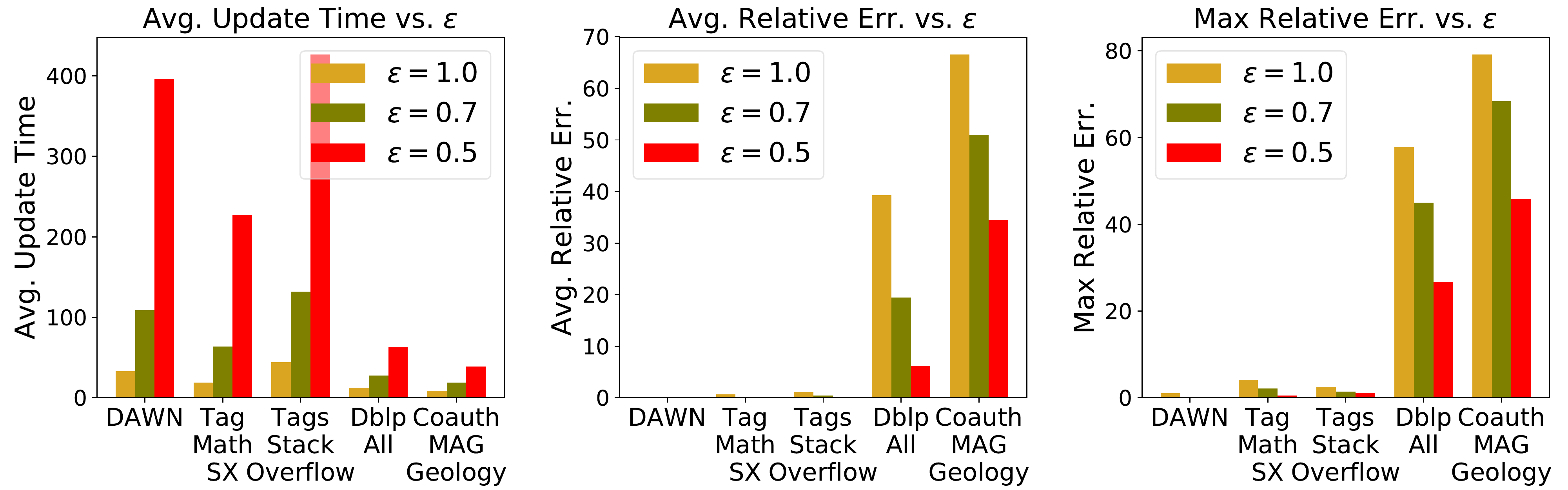}
    \caption{Accuracy vs Efficiency Trade-off for \udshp (for Unweighted Dynamic Hypergraphs): On the left, we plot the average update time for different settings of $\eps$. On the center and the right, we show the effect of $\eps$ on the average relative error and maximum relative error (over the reporting intervals).}
    \label{fig:dyn_unweighted_time_max_avg_density}
\end{figure}

We observe that for the hypergraphs (such as \dawn, \tagm, \tagsf) with high density values ($\Omega(\log n)$), the average and maximum relative errors are quite low ($<2-5\%$).
Thus, we recommend to use \udshp with larger values of $\epsilon$ (like  $\eps=1$) for these hypergraphs. Note that reduction in update time is quite dramatic ($\sim$8x) when increasing $\eps$ from 0.5 to 1.0 for these graphs.
For the hypergraphs with low density values ($o(\log n)$) the relative errors can go well above $30\%-40\%$ for larger values of $\eps$. 
Thus, we recommend using \udshp with smaller values of $\epsilon$ (like  $\eps=0.3$) for more accurate solutions, as for hypergraphs like \cmag, reducing $\eps$ from 1.0 to 0.5 reduces the average relative error from $70\%$ to $30\%$ (albeit at the expense of 3-fold increase in average update time).

\mypar{Weighted Hypergraphs}We give an account of our evaluations for weighted hypergraphs below.

\begin{figure*}[!h]
    \centering
    \includegraphics[width=\textwidth]{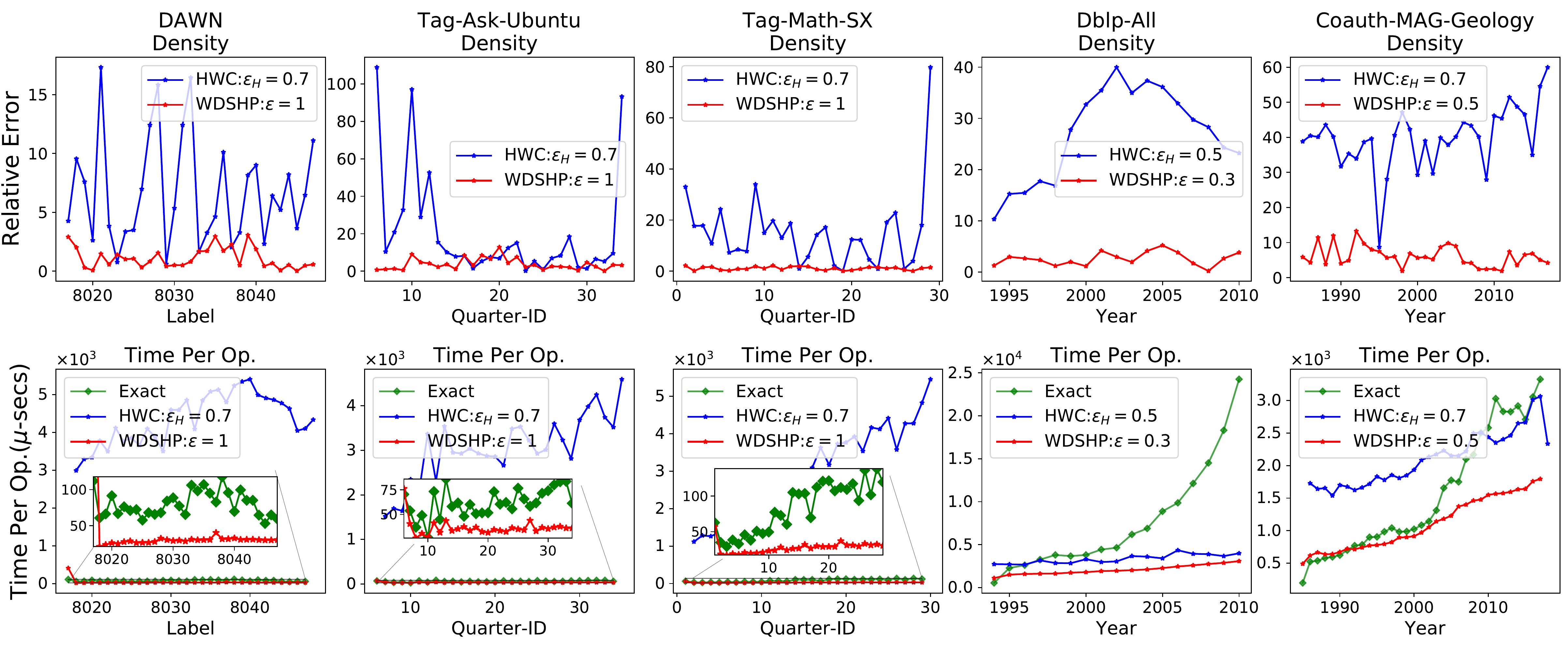}
    \caption{Accuracy and Efficiency Comparison for Weighted Dynamic Hypergraphs: The top row shows the relative error in the reported maximum density by \wdshp and \hwc with respect to \exact when run with the specified parameters. The bottom row plots the average update time taken by \wdshp, \hwc, and \exact for each reporting intervals. For each dataset (column), the parameter settings are identical.}
    \label{fig:dyn_weighted_error_time}
\end{figure*}
\mypar{Accuracy and Efficiency Comparison against Baselines}
    In~\Cref{fig:dyn_weighted_error_time}, we consider a similar setting as that of~\Cref{fig:dyn_unweighted_error_time} for the weighted case. In the top row, we compare the relative error percentage of \wdshp and \hwc. In the bottom row, we plot the average update times of \wdshp, \hwc, and \exact with same parameters (for each hypergraph). 
    
    We observe that, for most of the hypergraphs, the relative error of \hwc fluctuates quite a lot whereas that of the \wdshp remains quite stable. For \dblpall and \cmag, \wdshp outperforms \hwc quite remarkably in terms of accuracy. 
    In terms of average update time, \wdshp and \exact are comparable for smaller datasets, however \exact runs into scalability issue as the size increases (\cmag and \dblpall) and \wdshp easily outperforms \exact. \wdshp is remarkably better than \hwc in terms of update time  across all the datasets and over all the time points.
    This is perhaps expected, as in \hwc an edge with weight $w_e$ is processed by creating $w_e$ many copies of that edge.
    In \Cref{sec:addlexp}, we provide more results for different parameter settings (\Cref{fig:dyn_weighted_error_time_2_eps}).
    
    \begin{figure}
\centering
\includegraphics[scale=0.3]{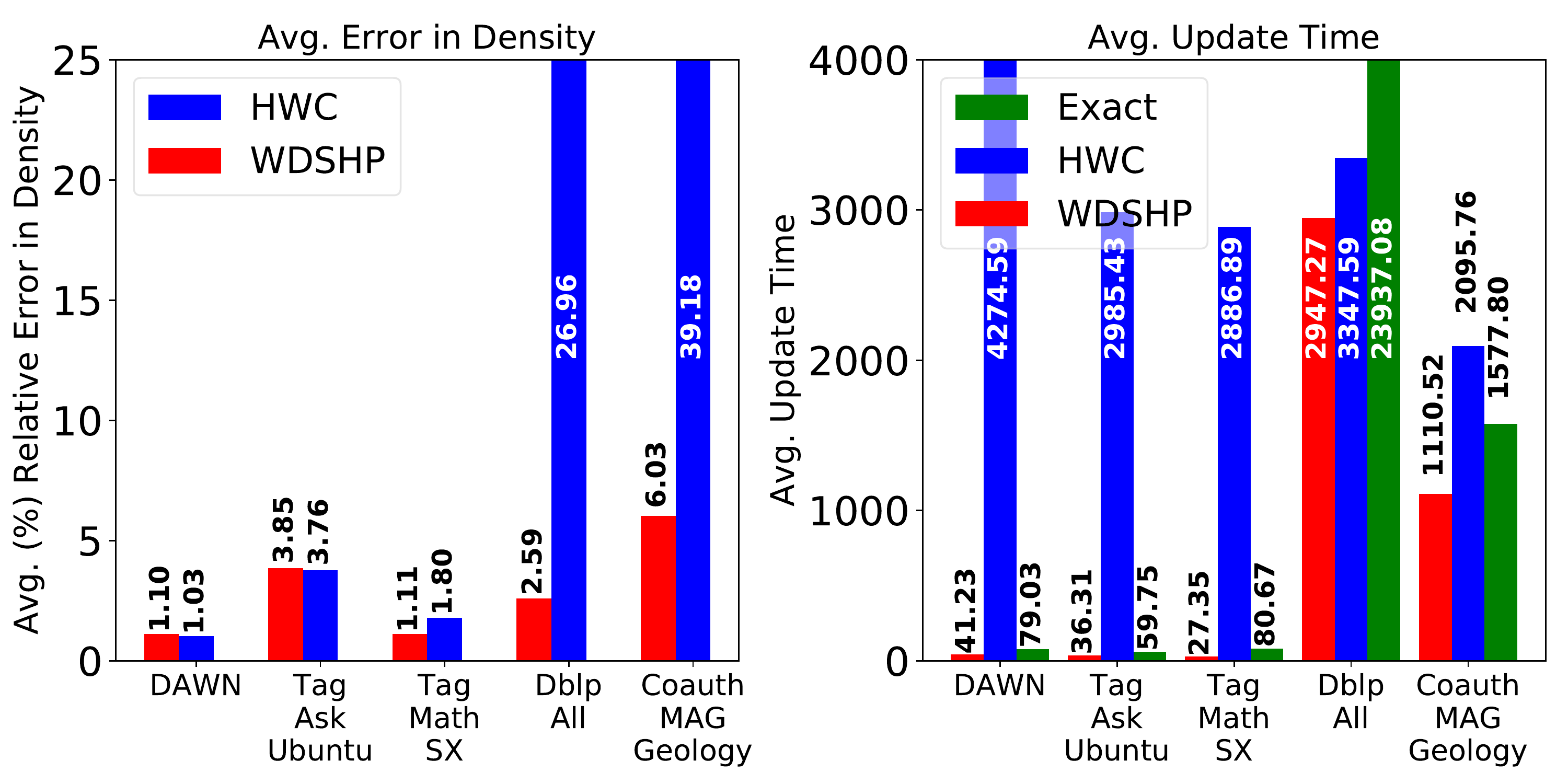}
\caption{Average Accuracy and Efficiency Comparison for fully dynamic weighted hypergraphs: on the left, we plot the average relative error of \wdshp and \hwc over all the reporting intervals for each dataset, and on the right, we compare the average update time of \wdshp, \hwc, and \exact over the entire duration.}
\label{fig:avg_time_error_dyn_weighted}
\end{figure}

    Analogous to the unweighted case, in~\Cref{fig:avg_time_error_dyn_weighted}, we quantify the average gain of \wdshp over the entire duration. Our main observations are as follows. (1) While for smaller datasets, \wdshp is 1.4x-3x faster on average compared to \exact, for larger datasets the gain can be as large as 8x (while maintaining $<6\%$ average error). This is remarkable considering that edge weights can be as large as 100. (2) For smaller hypergraphs (\dawn,\tagu,\tagm), to achieve a comparable accuracy of less than 5\%, \hwc is 100x times slower on average compared to \wdshp. (3) For larger datasets (\cmag, \dblpall), the average update times between \wdshp and \hwc are comparable; however \wdshp provides massive 80-90\% improvements over \hwc in terms of relative error. 
    
    \begin{figure}[h]
    \centering
        \includegraphics[scale=0.3]{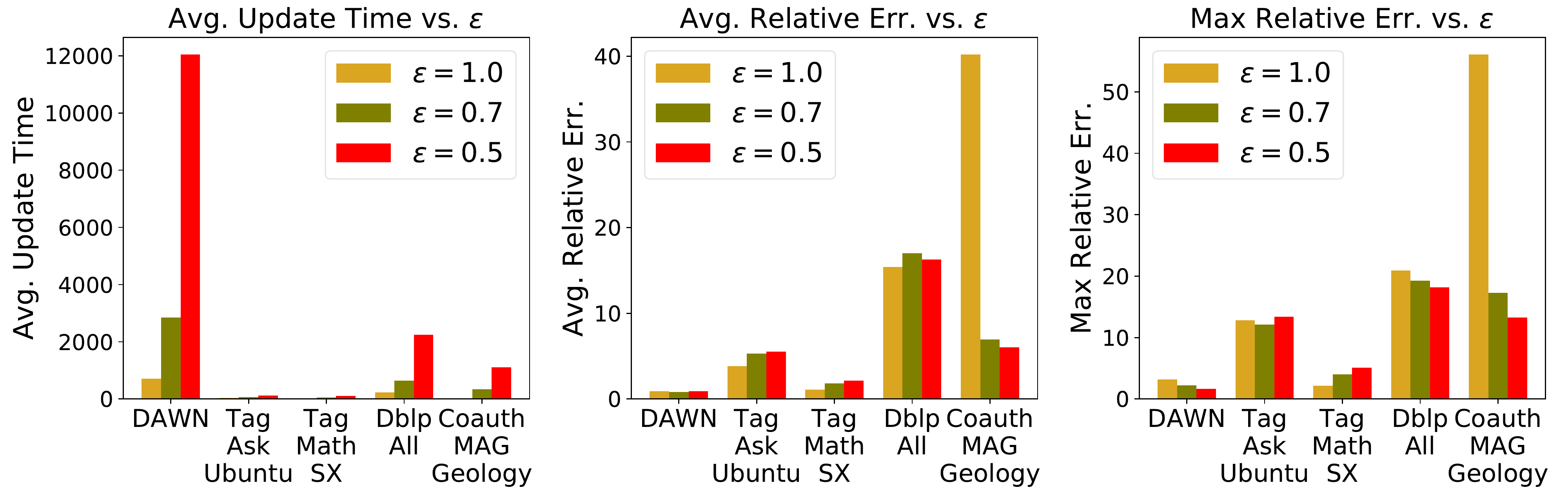}
        \caption{Accuracy vs Efficiency Trade-off for \wdshp (for Weighted Dynamic Hypergraphs): On the left, we plot the average update time for different settings of $\eps$. On the center and the right, we show the effect of $\eps$ on the average relative error and maximum relative error (over the reporting intervals).}
        \label{fig:dyn_weighted_time_max_avg_density}
    \end{figure}
    
    \mypar{Efficiency vs Accuracy trade offs for \wdshp}
    As similar to the unweighted setting,in~\Cref{fig:dyn_weighted_time_max_avg_density} we analyze the change in the {\em average update time} and the {\em average and max relative error} for \wdshp for different values of $\eps\in \{1.0, 0.7, 0.5\}$. 
    We observe that for hypergraphs (\dawn, \tagm, \tagu, \dblpall) with high density values ($\Omega(\log n)$), the relative error is less sensitive to the changes in $\eps$. 
    Addition of weights increases the density value of all the datasets, thus the relative error (both average and maximum) is not affected by a lot with changes in $\eps$ where as update time improves dramatically with increasing $\eps$.
    Thus, as similar to that in dynamic unweighted case, we recommend using \wdshp with higher values of $\eps$ (like $\eps=1$) for hypergraphs with high density values, and using low values of $\eps$ (like $\eps=0.3$) for hypergraphs with low density values ($o(\log n)$).

%% file: insertion-only-case.tex
\subsection{Insertion-only Case}
\label{subsec:insertion_only}
In this subsection, we consider the insertion-only hyperedge stream where the hyperedges are added following their timestamp. We perform extensive experiments with both unweighted and weighted edges. 

\mypar{Unweighted Hypergraphs} Let us first focus on the unweighted case.
\begin{figure*}[!ht]
    \centering
    \includegraphics[width=\textwidth]{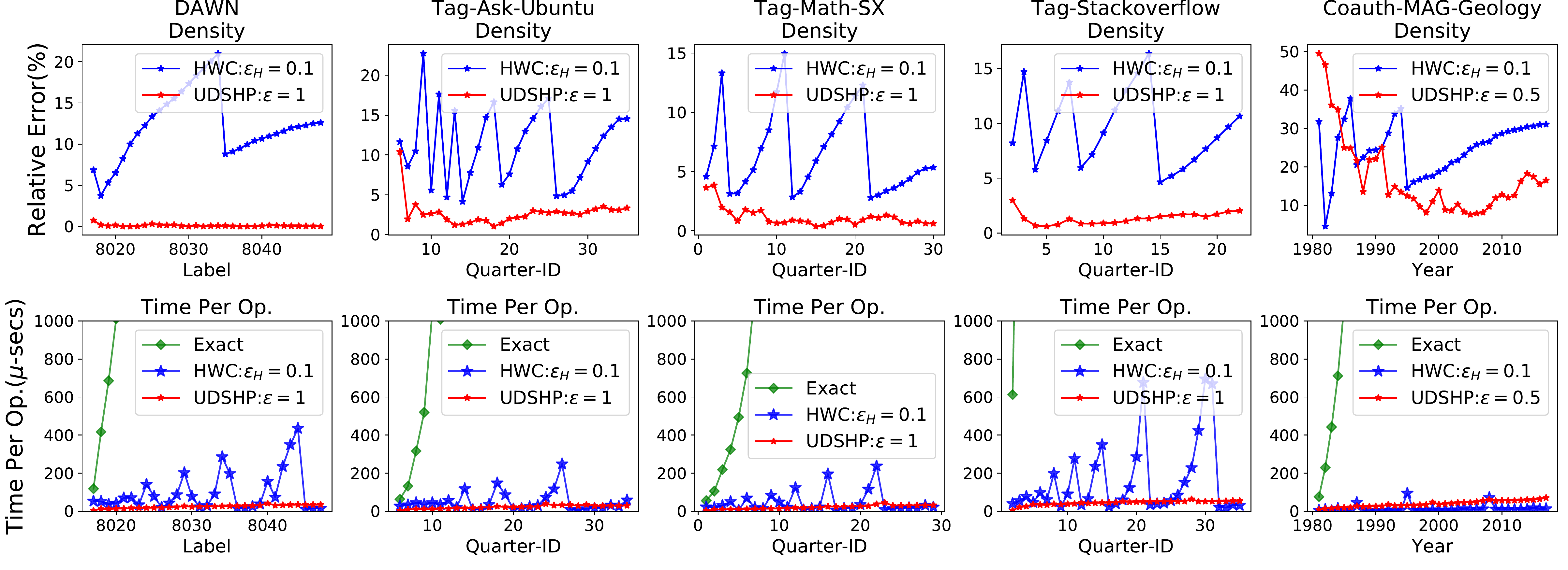}
    \caption{Accuracy and Efficiency Comparison for Unweighted Insertion-only Hypergraphs: The top row shows the relative error in the maximum density by \udshp and \hwc with respect to \exact when run with the specified parameters. The bottom row plots the average update time taken by \udshp, \hwc, and \exact for each reporting intervals. For each dataset (column), the parameter settings are identical.}
    \label{fig:inc_unweighted_error_speed}
\end{figure*}

\mypar{Accuracy and Efficiency Comparison against Baselines} Similar to the dynamic settings, 
    we first compare the accuracy of \udshp and \hwc with respect to \exact in~\Cref{fig:inc_unweighted_error_speed} top row. In the bottom row, we plot the average time taken per operation by \exact, \udshp, and \hwc during each reporting interval, maintaining the same parameter settings as that of top row. To further quantify the gain of \udshp, in~\Cref{fig:avg_time_error_inc_unweighted}, we compare the performances of \udshp against \hwc and \exact in terms of average relative error and average update time, where the average is taken over all the reporting intervals. We highlight some of our main findings below.

(1) Performance of \hwc fluctuates quite a lot over time as evident from the {\em saw-tooth} behaviour in the relative error and the update time curves for \hwc in~\Cref{fig:inc_unweighted_error_speed}. Thus, even when the average case update time for \hwc is low, the worst-case update time could be very high. In contrast, \udshp exhibits a much more stable behavior over time, making it more suitable practically. Note that this is consistent with the theoretical results for the respective algorithms since \hwc only guarantees small \emph{amortized} update time while \udshp guarantees small \emph{worst-case} update time.
In~\Cref{sec:addlexp}, we show results for more parameter settings (\Cref{fig:inc_unweighted_error_speed_2_eps}).
    
    \begin{figure}[H]
    \centering
    \includegraphics[scale=0.3]{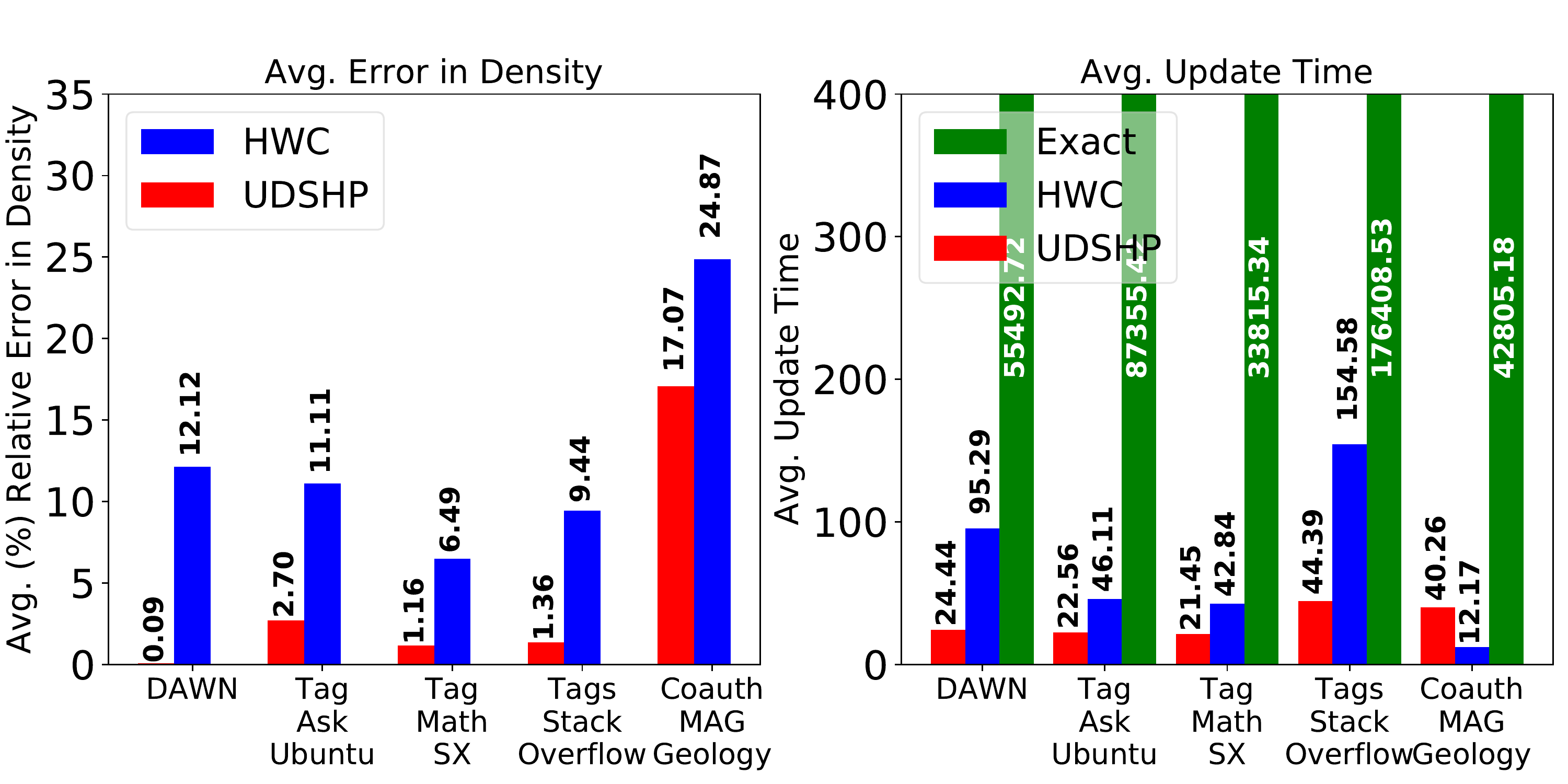}
    \caption{Average Accuracy and Efficiency Comparison for Unweighted Incremental Setting: On the left, we plot the average relative error of \udshp and \hwc, and on the right, we compare the average update time of \udshp, \hwc, and \exact. The average is taken over the entire duration.}
    \label{fig:avg_time_error_inc_unweighted}
\end{figure}
    (2) For the first four datasets, on average \udshp has 70\% better accuracy while being 2x-4x faster (on average) compared to \hwc (\Cref{fig:avg_time_error_inc_unweighted}). For the largest dataset \cmag, \hwc indeed has an edge over \udshp in terms of average update time while both incurring comparable loss in accuracy (\Cref{fig:avg_time_error_inc_unweighted}). However, as we noted before, the {\em saw-tooth} behavior of \hwc implies a higher worst-case update time for \hwc compared to \udshp (\Cref{fig:inc_unweighted_error_speed}).
    
    (3) \exact performs extremely poorly in the incremental settings, as one would expect. The sizes of the hypergraphs are much larger compared to the dynamic settings, making \exact extremely unsuitable for any practical purpose.
    
    \mypar{Efficiency vs Accuracy trade offs} As similar to that in the dynamic setting, in ~\Cref{fig:inc_unweighted_time_max_avg_error}, we analyze the change in the {\em average update time} and the {\em average and max relative error} for \udshp for different values of $\eps\in \{1.0, 0.7, 0.5\}$. 
    We observe that even if the update time is sensitive to change in $\eps$, the average and maximum relative error for all the high density ($\Omega(\log n)$) hypergraphs (\dawn, \tagu, \tagm, \tagsf) is low ($<10\%$). 
    And thus we recommend, using \udshp with high value of $\eps$ (like $\eps=1$) for these hypergraphs.
    On the other hand for the low density ($o(\log n)$) hypergraphs (like \cmag), we recommend using \udshp with low value of $\eps$ (like $\eps=0.5$ or $0.3$).
    \begin{figure}[!ht]
    \centering
    \includegraphics[scale=0.3]{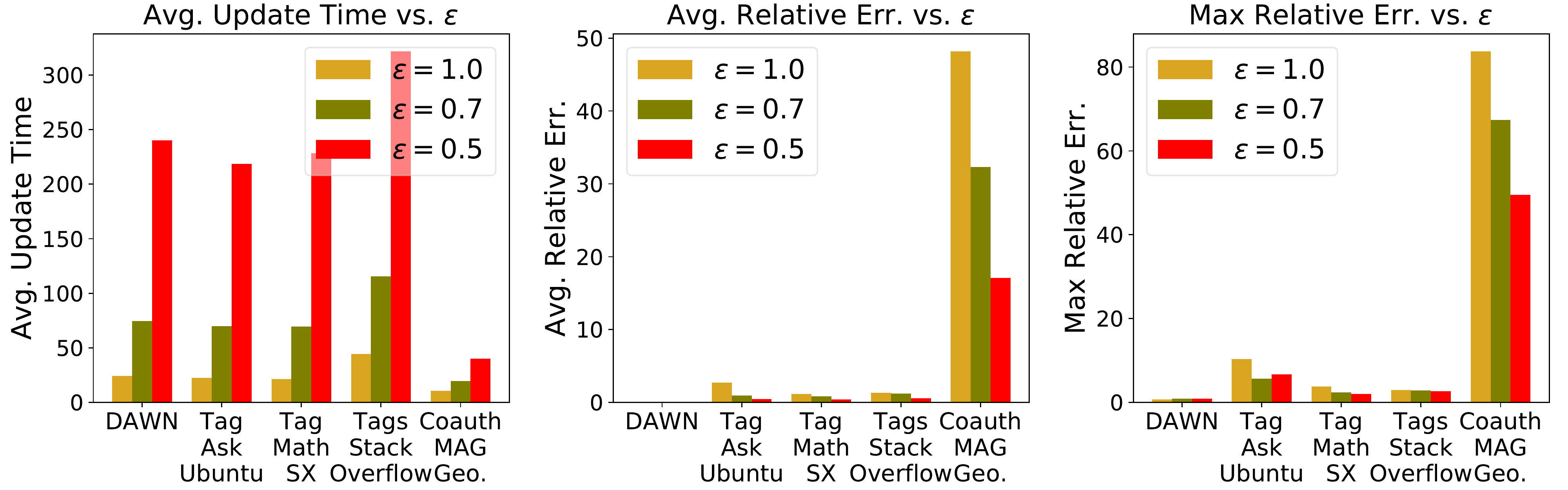}
    \caption{Accuracy vs Efficiency Trade-off for Unweighted Incremental Hypergraphs (\udshp): We plot the average update time (left), average relative error (middle), and maximum relative error (right) over the reporting intervals for different settings of $\eps$.}
    \label{fig:inc_unweighted_time_max_avg_error}
\end{figure}

\mypar{Weighted Hypergraphs}
\label{sec:incremental_weighted_expts}
We now turn to our experimental evaluations for the weighted case.

\begin{figure}[!ht]
    \centering
    \includegraphics[width=\textwidth]{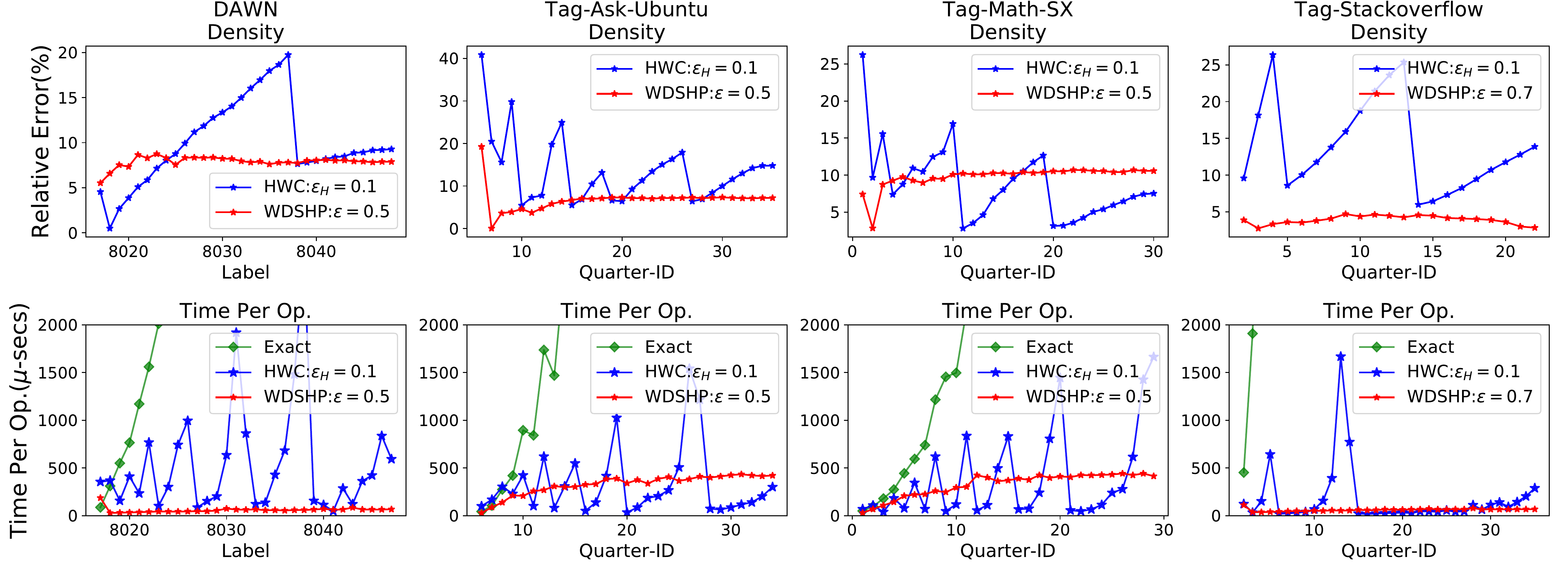}
    \caption{Accuracy and Efficiency Comparison for Weighted Insertion-only Hypergraphs: The top row shows the relative error in the maximum density by \wdshp and \hwc with respect to \exact when run with the specified parameters. The bottom row plots the average update time taken by \wdshp, \hwc, and \exact for each reporting intervals. For each dataset (column), the parameter settings are identical.}
    \label{fig:inc_weighted_error_time}
\end{figure}
\begin{figure}[!ht]
    \centering
    \includegraphics[scale=0.3]{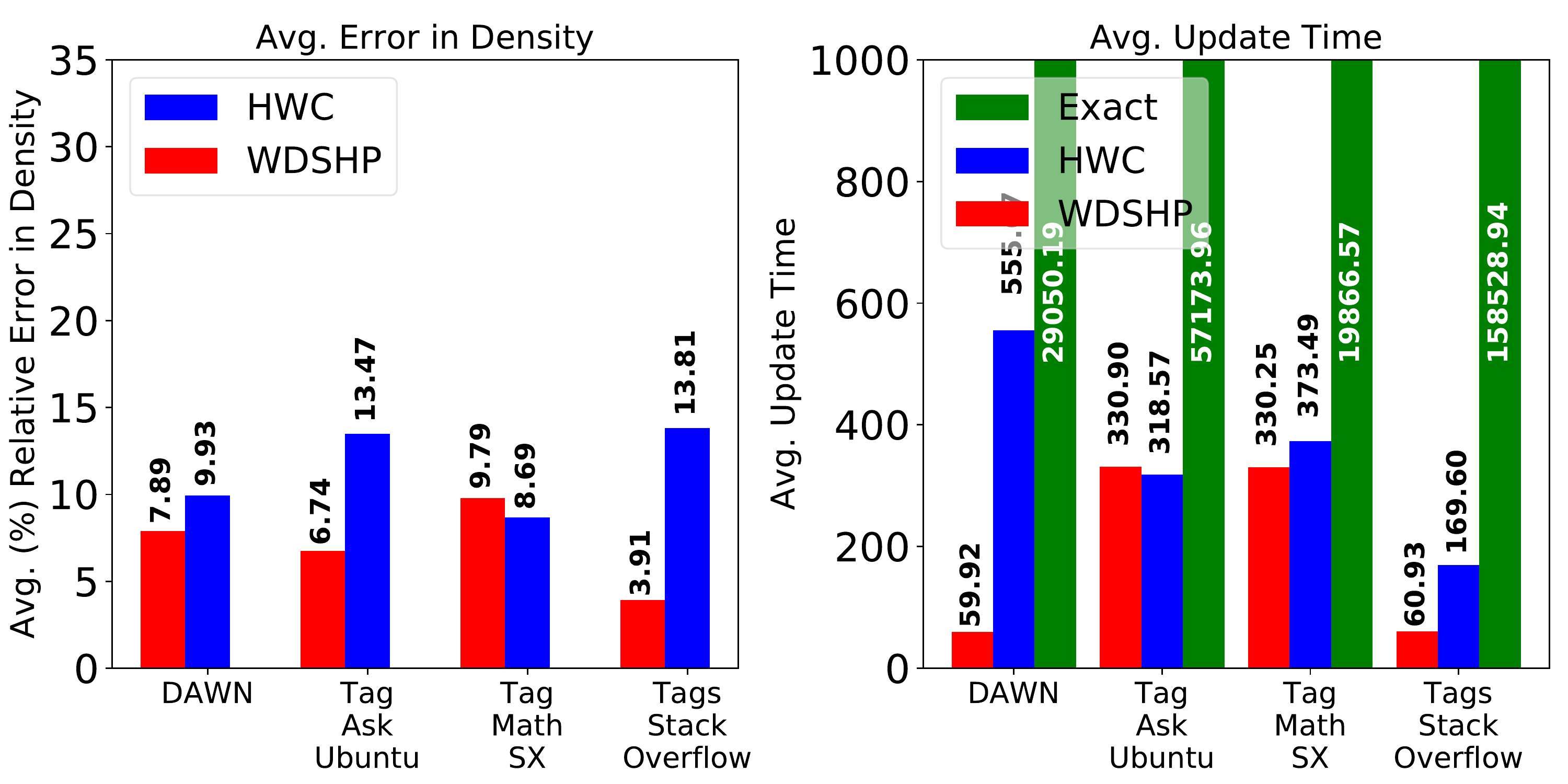}
    \caption{Average Accuracy and Efficiency Comparison for Weighted Incremental Setting: On the left, we plot the average relative error of \wdshp and \hwc, and on the right, we compare the average update time of \wdshp, \hwc, and \exact. The average is taken over the entire duration.}
    \label{fig:avg_time_error_inc_weighted}
\end{figure}
\vspace{-0.5mm}
\mypar{Accuracy and Efficiency Comparison against Baselines}
    Similar to the incremental unweighted setting, here as well we first compare the accuracy of \wdshp and \hwc with respect to \exact in~\Cref{fig:inc_weighted_error_time} top row. In the bottom row, we plot the average time taken per operation by \exact, \wdshp, and \hwc during each reporting interval, maintaining the same parameter settings as that of top row. 
    In \Cref{sec:addlexp}, we show results for some more parameter settings (\Cref{fig:inc_weighted_error_time_2_eps}).
    
    To further quantify the gain of \udshp, in~\Cref{fig:avg_time_error_inc_weighted}, we compare the performances of \udshp against \hwc and \exact in terms of average relative error and average update time, where the average is taken over all the reporting intervals. Following are some of our main findings:
    
   
    (1) Performance of \hwc fluctuates quite a lot over time as evident from the {\em saw-tooth} behaviour in the relative error and the update time curves for \hwc in~\Cref{fig:inc_weighted_error_time}. 
    Thus, even when the average case update time for \hwc is low, the worst case update time could be very high. 
    In contrast, \udshp exhibits a much more stable behavior over time, making it more suitable for practical use.
    Additionally, as the size of the hypergraph increases, we observe that the worst case update time for \hwc increases significantly as well, as compared to that of \wdshp.
    
    (2) As similar to that of the incremental unweighted setting, here as well, \exact performs extremely poorly as expected.
    And the similar reason holds here, i.e. the sizes of the hypergraphs are much larger compared to the dynamic settings, making \exact extremely unsuitable for any practical purpose.
    
    (3) Across all these datasets, \wdshp performs better than \hwc in terms of average relative error except for the \tagm dataset, where \hwc is around 12\% better.
    And in terms of average update time, \wdshp is either comparable to or significantly outperforms \hwc.
    
\begin{figure}[h]
    \centering
    \includegraphics[scale=0.3]{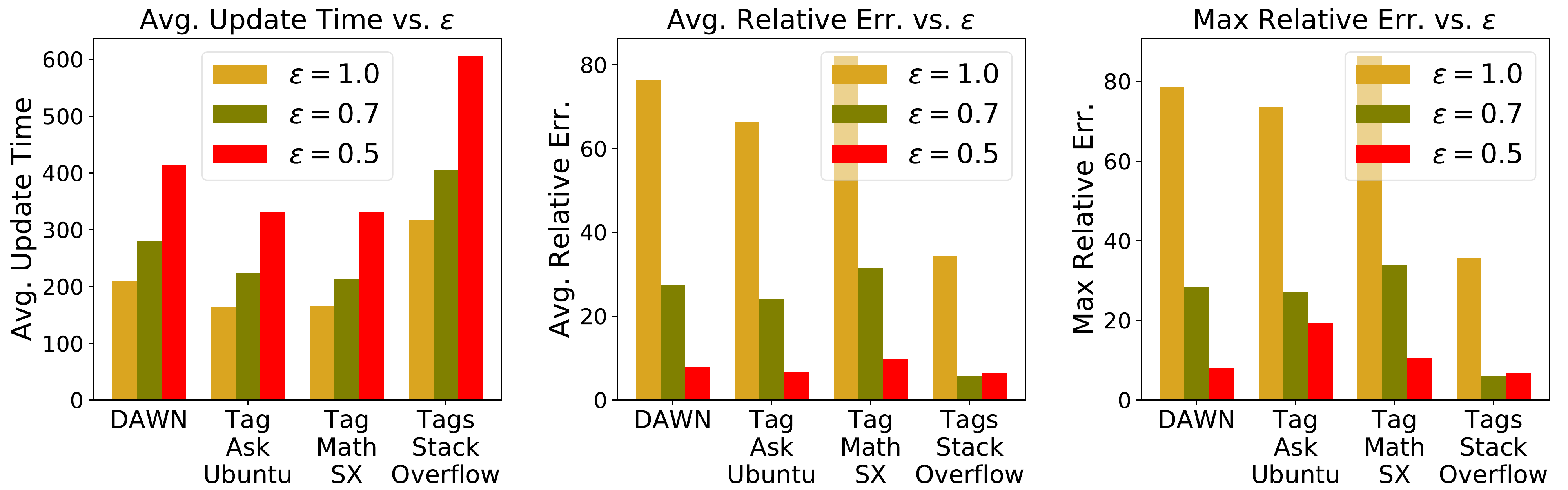}
    \caption{Accuracy vs Efficiency Trade-off: We give trade-offs for \wdshp for Weighted Insertion only Hypergraphs: On the left, we plot the average update time for different settings of $\eps$. In the middle and right, we show the effect of $\eps$ on the average relative error and maximum relative error (over the reporting intervals).}
    \label{fig:inc_weighted_time_max_avg_error}
\end{figure}

\mypar{Efficiency vs Accuracy trade offs for \wdshp}
In~\Cref{fig:inc_weighted_time_max_avg_error}, we show the performance of \wdshp in terms of average update time and average and maximum relative error incurred for different values of $\epsilon \in \{0.5, 0.7, 1.0\}$.
Here we observe that the gain in relative error is high as compared to that in average update time for change in $\epsilon$ from $1$ to $0.5$.  
For example, for the \dawn dataset, changing $epsilon$ from $1$ to $0.5$ increases the average update time by 2x times, whereas, the improvement in average and maximum relative error is $\sim80\%$.

%% file: conclusions.tex
\section{Conclusions}
\label{sec:conclusions}
In this work, we give a new dynamic algorithm for the
densest subhypergraph problem that maintains an $(1+\eps)$-approximate 
solution with $O(\text{poly}(r,\log n,\eps^{-1}))$ worst-case update time. 
Our approximation ratio is independent of the rank $r$, a significant improvement 
over previous works. 
Through extensive experiments, we analyze the accuracy and efficiency of our algorithm and find that it significantly outperforms the state of the art ~\cite{hu2017maintaining} both in terms of accuracy and efficiency. While our update time is near-optimal in terms of $r$, the asymptotic dependency on $\eps^{-1}$ is rather larger. We leave the question of improving this dependency as an interesting future direction.

%% file: additional_expts.tex
\section{Additional Experimental Results}\label{sec:addlexp}

In this section, we provide some additional experimental results for various settings of the parameters $\eps$ and $\eps_H$. They follow similar trends as described for each case in \Cref{sec:exp}.

\mypar{Fully Dynamic Setting: Unweighted and Weighted}
\begin{figure*}[!ht]
    \centering
    \includegraphics[width=\textwidth]{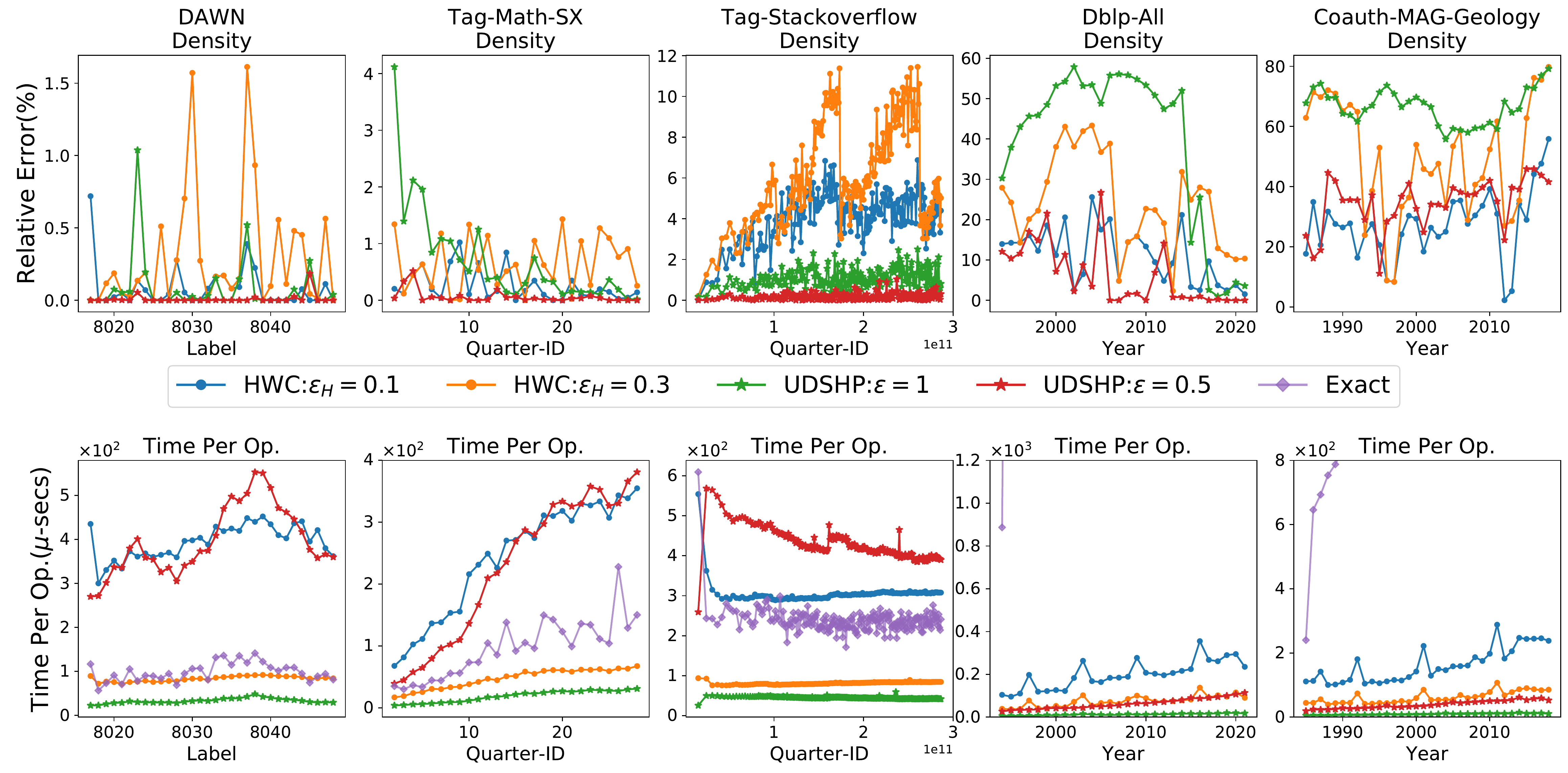}
    \caption{Accuracy and Efficiency Comparison for Unweighted Dynamic Hypergraphs: The top row shows the relative error in the reported maximum density by \udshp and \hwc with respect to \exact when run with the specified parameters. The bottom row plots the average update time taken by \udshp, \hwc, and \exact for each reporting intervals. For each dataset (column), the parameter settings are identical.}
    \label{fig:dyn_unweighted_error_time_2_eps}
\end{figure*}

\begin{figure}[H]
    \centering
    \includegraphics[width=\textwidth]{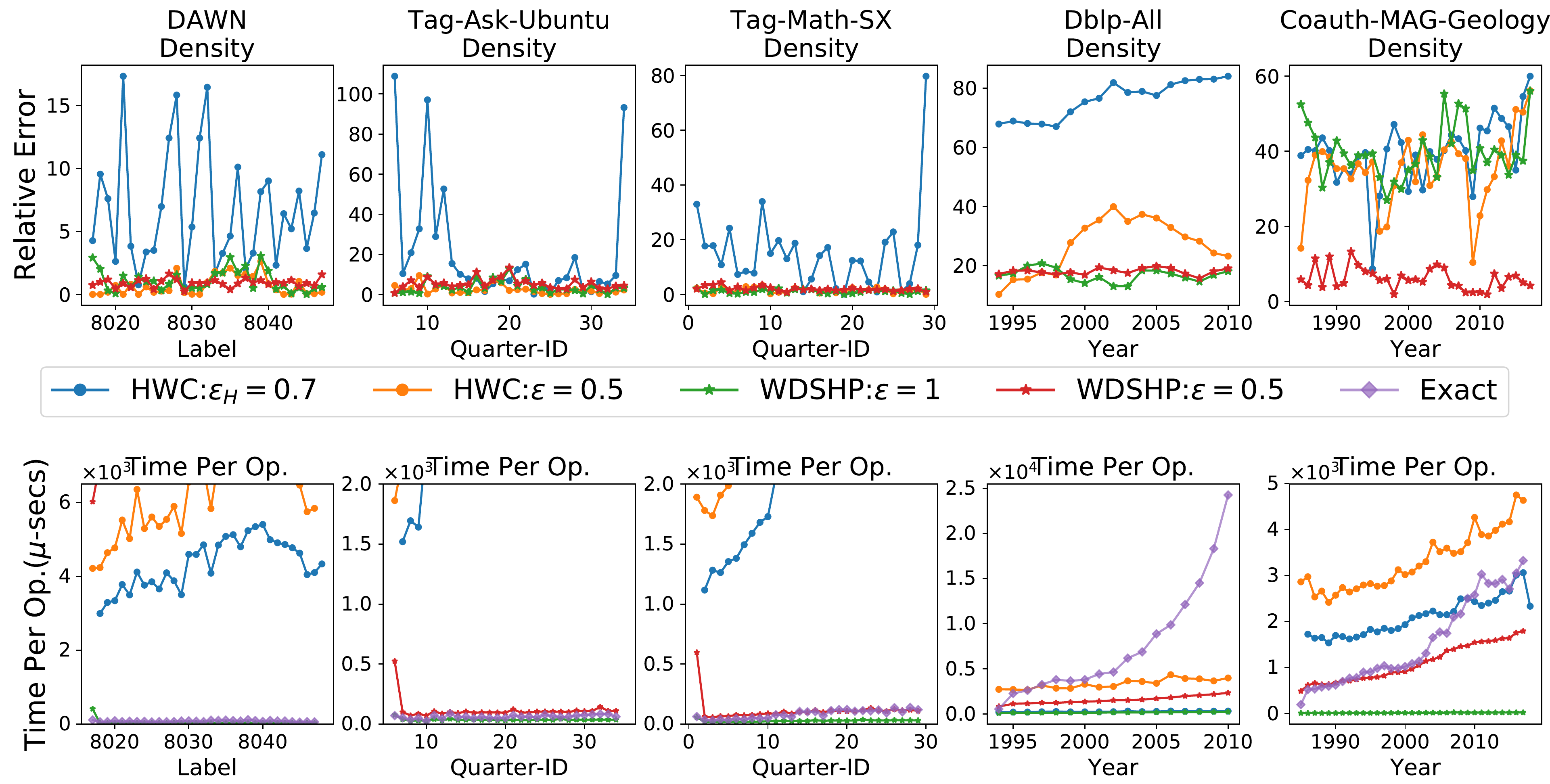}
    \caption{Accuracy and Efficiency Comparison for Weighted Dynamic Hypergraphs: The top row shows the relative error in the reported maximum density by \udshp and \hwc with respect to \exact when run with the specified parameters. The bottom row plots the average update time taken by \udshp, \hwc, and \exact for each reporting intervals. For each dataset (column), the parameter settings are identical.}
    \label{fig:dyn_weighted_error_time_2_eps}
\end{figure}

\mypar{Insert-only Setting: Unweighted and Weighted}
\begin{figure}[H]
    \centering
    \includegraphics[width=\textwidth]{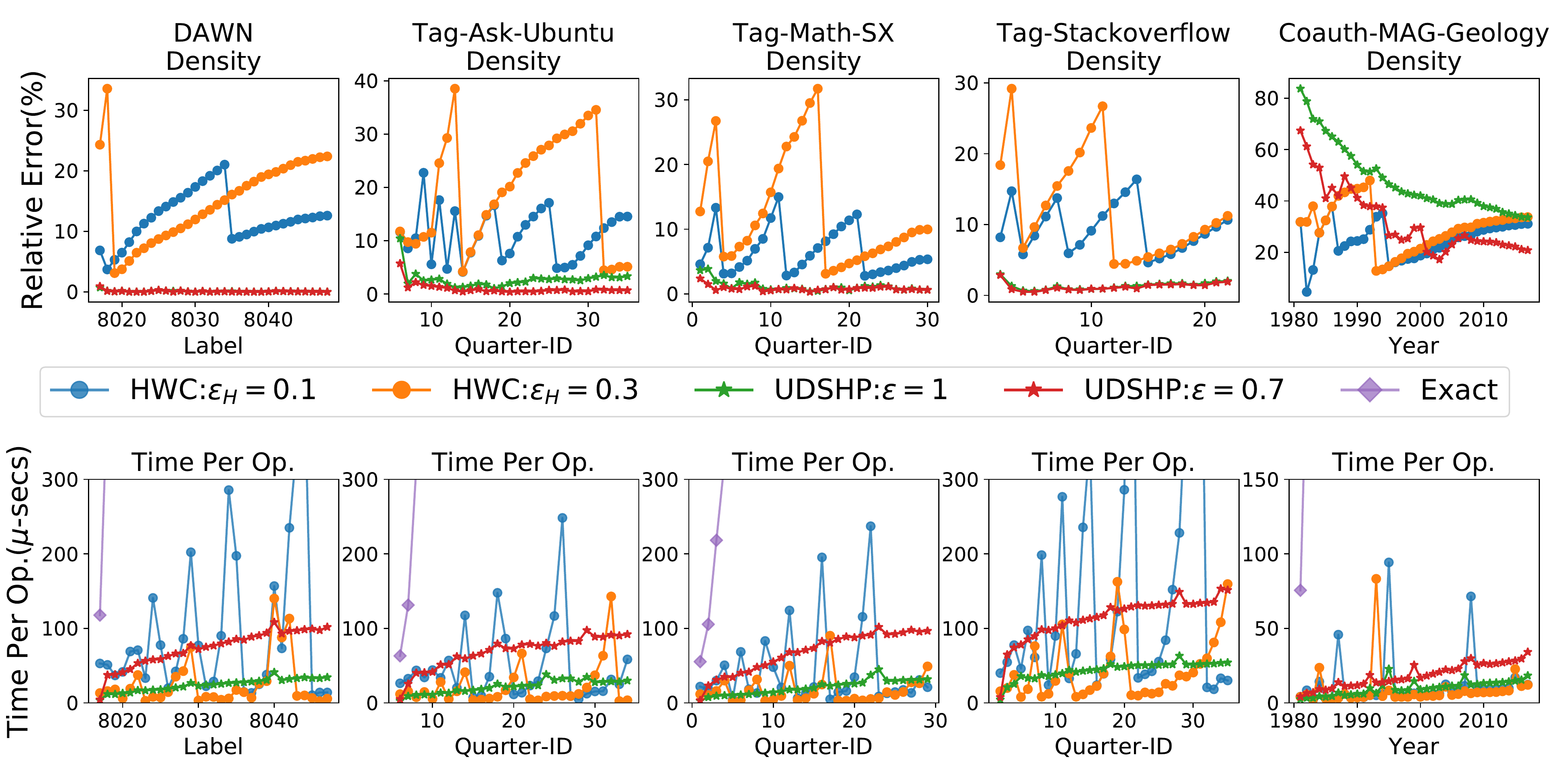}
    \caption{Accuracy and Efficiency Comparison for Unweighted Insertion-only Hypergraphs: The top row shows the relative error in the maximum density by \udshp and \hwc with respect to \exact when run with the specified parameters. The bottom row plots the average update time taken by \udshp, \hwc, and \exact for each reporting intervals. For each dataset (column), the parameter settings are identical.}
    \label{fig:inc_unweighted_error_speed_2_eps}
\end{figure}

\begin{figure}[H]
    \centering
    \includegraphics[width=\textwidth]{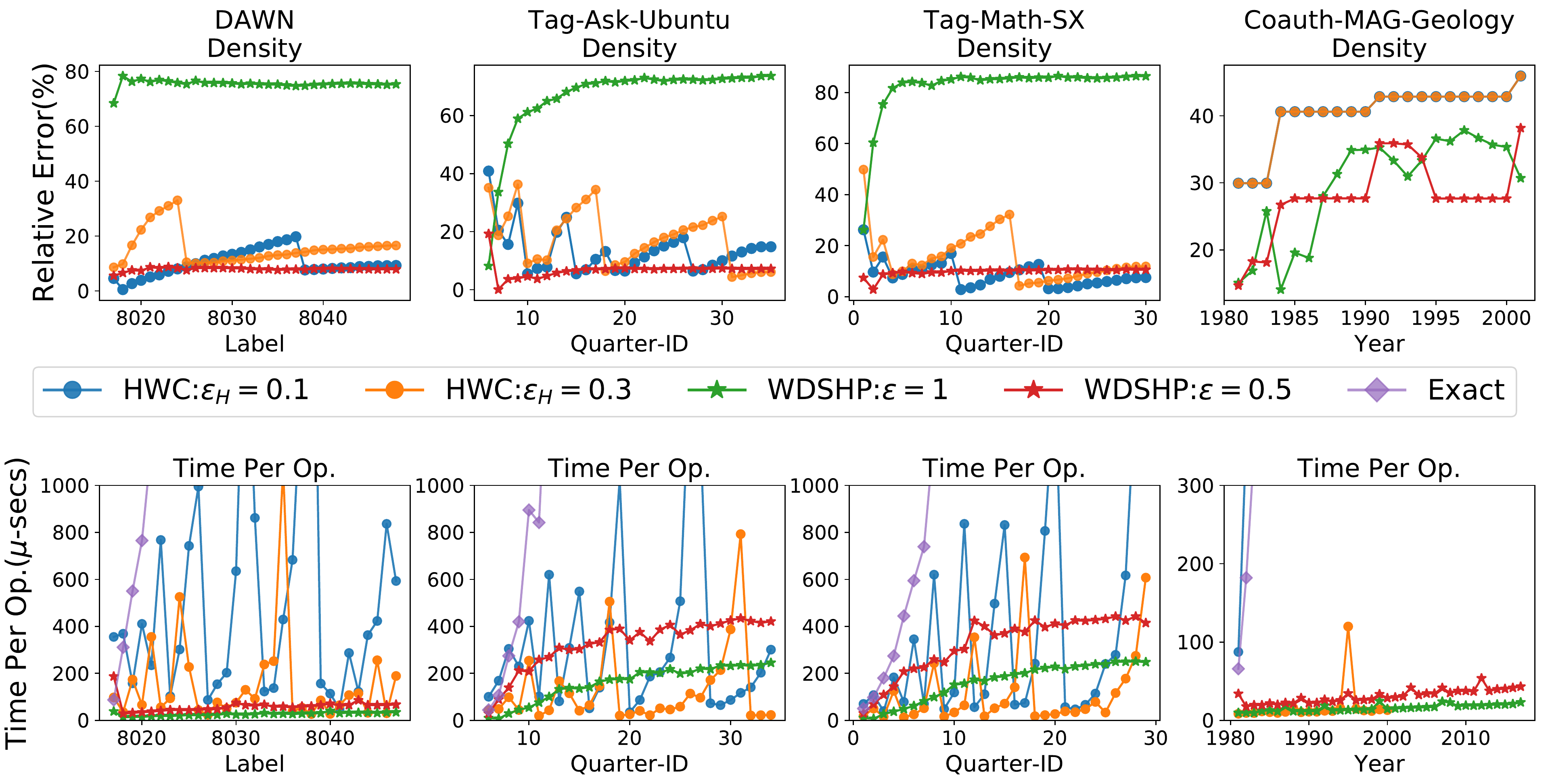}
    \caption{Accuracy and Efficiency Comparison for Weighted Insertion-only Hypergraphs: The top row shows the relative error in the maximum density by \wdshp and \hwc with respect to \exact when run with the specified parameters. The bottom row plots the average update time taken by \wdshp, \hwc, and \exact for each reporting intervals. For each dataset (column), the parameter settings are identical.}
    \label{fig:inc_weighted_error_time_2_eps}
\end{figure}